\documentclass{article}   	
\usepackage{amsthm,amsmath,amssymb}
\usepackage{mathtools}
\theoremstyle{plain}
\newtheorem{theorem}{Theorem}[section]
\newtheorem{lemma}[theorem]{Lemma}
\newtheorem{proposition}[theorem]{Proposition}
\newtheorem{corollary}[theorem]{Corollary}

\theoremstyle{definition}
\newtheorem{definition}[theorem]{Definition}

\theoremstyle{remark}
\newtheorem{remark}[theorem]{Remark}

\newcommand{\pba}{\mathsf{pba}}

\DeclareMathOperator{\spec}{spec}
\newcommand{\Outloc}[2]{\mathrm{Out}^{\mathrm{loc}}_{#1}(#2)}
\newcommand{\Outglob}[2]{\mathrm{Out}^{\mathrm{glob}}_{#1}(#2)}

\usepackage{geometry}                		
\usepackage{graphicx}
\usepackage{cite}
\providecommand{\bigsqcap}{\mathop{\raisebox{-0.25ex}{\scalebox{1.45}{$\sqcap$}}}}

\usepackage[dvipsnames,svgnames,table]{xcolor}	

\title{Quantum uncertainty in a macroscopic domain}

\author{Othman Q. Malhas$^*$
\and
Bacim Alali$^\dag$}

\date{}				

\begin{document}
\maketitle

\begin{abstract}
We develop a classical model-theoretic representation of partial Boolean algebras and use it to formulate quantum-like uncertainty without replacing classical propositional logic.  Given a surjection from the initial formul\ae\ of a propositional language onto a partial Boolean algebra $(\mathcal V,\Pi)$, we construct a consistent theory $\mathcal T_g$ whose core---the ordered set of equivalence classes of initial formul\ae---is isomorphic to $(\mathcal V,\Pi)$.  The models of $\mathcal T_g$ are characterized as upward-closed clusters, and this characterization yields a model-theoretic formulation of KS-colourability: an $n$-dimensional partial Boolean algebra admits a KS-colouring exactly when the induced theory has a model meeting every pre-frame in one primitive formula.

For finite-spectrum observables, we define uncertainty by the number of locally admissible atomic outcomes.  A model is dispersion-free precisely when it meets every pre-frame in one primitive formula; equivalently, its core is KS-colourable.  After introducing a measurement-update rule as an additional modelling postulate, we show in a concrete example that measurement of an incompatible observable can destroy sharpness, giving a finite model-theoretic form of back-action.  Two finite structures illustrate the distinction: the $12$-vertex partial Boolean algebra is KS-colourable, whereas a rigorously constructed $140$-vertex, four-dimensional partial Boolean algebra is not, as shown by a parity argument.  The underlying incidence structure of the latter is isomorphic to the Peres $24$-ray, $24$-basis configuration in $\mathbb R^4$, so the example is simultaneously combinatorial and Hilbert-space realizable.  Both structures may be assigned macroscopic interpretations, demonstrating that the formal phenomena arise from the organization of propositions and models rather than from nonclassical deduction itself.  Finally, we relate certain models to the probability-one propositions of pure states and density operators, while emphasizing that such certainty models do not determine the full quantum state.
\end{abstract}

\section{Introduction}

The propositional structure associated with a quantum system is generally non-Boolean.  In the standard Hilbert-space formulation, closed subspaces, or equivalently orthogonal projections, form compatible Boolean algebras only locally; globally, distributivity fails and not every pair of propositions belongs to a common Boolean context.  This observation motivated the development of quantum logic and remains central to the study of incompatibility and contextuality \cite{Varadarajan,KS}.

The failure of global Boolean structure does not, however, force a rejection of classical propositional deduction.  A classical theory has a Boolean Lindenbaum algebra, but the equivalence classes represented by its initial formul\ae\ need not form a Boolean subalgebra.  They form what we call the \emph{core} of the theory.  The central structural result of this paper is that every partial Boolean algebra can be realized, up to isomorphism, as the core of a consistent theory formulated entirely in classical propositional logic.  The non-Boolean structure is therefore located in the distinguished family of propositions and their compatibility relations, not in a modification of the classical truth tables or consequence relation.

More precisely, let $(\mathcal V,\Pi)$ be a partial Boolean algebra and let $g:U\twoheadrightarrow\mathcal V$ map the initial formul\ae\ of a propositional language onto its vertices.  From order and orthogonality in $(\mathcal V,\Pi)$ we define a theory $\mathcal T_g$.  We prove that $\mathcal T_g$ is consistent, characterize all of its models, and show that its core is isomorphic to $(\mathcal V,\Pi)$.  The model characterization is expressed in terms of \emph{clusters}: sets of initial formul\ae\ containing no orthogonal pair.  Every cluster generates an upward-closed model, and every model is an upward-closed cluster.  This gives a direct bridge between the combinatorics of the partial Boolean algebra and the ordinary two-valued semantics of classical propositional logic.

That bridge permits a model-theoretic treatment of contextuality and uncertainty.  For an atomic finite-dimensional partial Boolean algebra, a pre-frame is a set of primitive formul\ae\ mapped onto a frame of mutually orthogonal atoms.  We use the term \emph{KS-colourable} for the existence of a $0$--$1$ assignment selecting exactly one atom from every frame.  We prove that KS-colourability is equivalent to the existence of a model meeting every pre-frame in exactly one primitive formula.  Such models are called \emph{dispersion-free}.  Thus the obstruction usually associated with the Kochen--Specker phenomenon appears here as the nonexistence of a particular kind of classical model, even though the ambient deductive system remains classical.

We then introduce finite-spectrum observables and intrinsic formul\ae\ of the form $(P,E)$, interpreted as ``the value of $P$ lies in $E$.''  Relative to a model, the locally admissible atomic outcomes of $P$ determine a finite measure of uncertainty.  Sharpness means that only one atomic outcome is locally admissible.  This definition does not require a probability measure and should therefore be understood as a possibilistic or certainty-based notion of uncertainty.  A model may also leave an observable with \emph{no} locally admissible outcome; we isolate this degenerate case explicitly, since the count of admissible outcomes measures uncertainty only when at least one outcome survives.  We additionally study dynamics by adopting a natural measurement-update rule that replaces the current model by the model generated by the observed atom.  The update rule is a modelling postulate, not a consequence of propositional logic alone; once adopted, however, its consequences are fixed by the cluster structure.  In the $12$-vertex example, a sharp observable remains sharp under immediate repetition, while measurement of an incompatible observable can render it nonsharp.  This gives an explicit finite example of model-theoretic back-action.

Two finite partial Boolean algebras organize the main examples.  The first, $\mathcal V_{12}$, is three-dimensional and KS-colourable.  It supports dispersion-free models and illustrates uncertainty that can be removed by refining to a suitable global valuation.  The second, $\mathcal V_{140}$, is a four-dimensional partial Boolean algebra constructed from a $3\times3$ array of modules.  We prove that its quotient operations are well defined, count its $140$ vertices, identify its frames, and establish by a parity argument that it is not KS-colourable.  Consequently, the induced theory has no dispersion-free model, and indeed no model in which every finite-spectrum observable is sharp.  The underlying $24$-point, $24$-block incidence structure turns out to be isomorphic to the Peres $24$-ray configuration in $\mathbb R^{4}$ \cite{Peres1991}, and the nine loops used in the parity argument reproduce the $18$-vector, $9$-basis proof of Cabello, Estebaranz and Garc\'ia-Alcaine \cite{Cabello1996}; the contribution here is the model-theoretic reading of that obstruction, not the configuration itself.  The initial formul\ae\ in either example may be given macroscopic interpretations.  These interpretations are illustrative: they show that the formal pattern depends on the organization of propositions and contexts, but they do not by themselves establish a physical violation of a probabilistic or temporal inequality.

The paper also clarifies the relation between its two-valued models and quantum states.  A model associated with an atom records exactly the propositions having probability one in the corresponding pure state.  More generally, a model associated with a nonzero support element records the probability-one propositions common to all density operators with that support.  We therefore call these \emph{certainty models}; they encode the support or certainty content of a state, not its complete probability distribution.  Many models of $\mathcal T_g$ are neither pure-state nor support-state certainty models.

This work continues a programme initiated in \cite{Malhas1987}, developed through the representation of quantum-logical structures as the poset or core of a classical theory in \cite{malhas1992quantum,mal2}, and extended model-theoretically in \cite{Malhas1995}.  The new contributions here are the systematic cluster characterization of models, the model-theoretic equivalence with KS-colourability, the finite-spectrum formulation of uncertainty, the measurement-update analysis, the dispersion-free criterion, the theorem that a model making \emph{every} observable sharp forces KS-colourability, and the complete construction and parity analysis of $\mathcal V_{140}$.  The connection with contextuality as a resource for quantum information \cite{Cabello2014,Howard2014} is structural at this stage: the present framework represents contextual compatibility patterns, but it does not yet supply probabilities, a resource monotone, or a computational speed-up.

The paper is organized as follows.  Sections~\ref{sec:basics}--\ref{sec:pba} develop the logical and partial-Boolean-algebraic foundations.  Sections~\ref{sec:simple-pba}--\ref{sec:compare} construct $\mathcal T_g$, characterize its models, and identify its core.  Section~\ref{sec:finite-pba} introduces finite-dimensional structures, Section~\ref{pm} discusses certainty models, and Section~\ref{ksp} gives the model-theoretic KS criterion.  Sections~\ref{sec:obs}--\ref{fuzzy} introduce observables, intrinsic formul\ae, macroscopic interpretations, and uncertainty.  Measurement dynamics and dispersion-free models are treated in Sections~\ref{sec:dynamics} and~\ref{sec:hv}.  Section~\ref{sec:4dim} constructs $\mathcal V_{140}$ and proves the parity obstruction, and Section~\ref{sec:outlook} records the interpretive limits and open directions.  The appendices relate certainty models to quantum states, describe the optional single-module example, and provide a finite verification certificate for $\mathcal V_{140}$.

\section{Preliminaries, terminology, and notation}
\label{sec:basics}

Much of the material in this section can be found in elementary textbooks on mathematical logic; see, for example \cite{Enderton}.
It is included here primarily to fix notation and terminology, and to introduce the central concept of this paper, namely the concept of \emph{the core of a theory} (see Section~\ref{sec:core}).

Let $U$ be a \emph{nonempty} set of symbols, whose members will be referred to as \emph{initial formul\ae}. 
Intuitively, the initial formul\ae\ are the building blocks of a propositional language; they represent the \emph{simplest meaningful assertions} within a given field of study.

We define the \emph{set of all well-formed formul\ae}, or simply the \emph{set of all formul\ae}, to be the \emph{smallest} set $\Gamma$ of finite strings of symbols satisfying the following conditions:
\begin{enumerate}
\item $U \subseteq \Gamma$;
\item If $A \in \Gamma$, then $(\neg A) \in \Gamma$;
\item If $A \in \Gamma$ and $B \in \Gamma$, then $(A \vee B)$, $(A \wedge B)$, $(A \to B)$, and $(A \leftrightarrow B)$ belong to $\Gamma$.
\end{enumerate}

Thus, each member of $\Gamma$ is a \emph{finite} sequence of symbols constructed from a finite number of applications of logical connectives and parentheses to elements of $U$. 
It is \emph{not} assumed that the set $U$ of initial formul\ae\ itself is finite.
Parentheses may be omitted whenever no ambiguity arises.
Initial formul\ae\ will be denoted by lowercase English letters $p, q, r, \dots$.

For any subset $V \subseteq U$, the \emph{characteristic function} of $V$ is the function $\chi_V : U \to \{0,1\}$ defined by
\[
\chi_V(p) =
\begin{cases}
1, & p \in V, \\
0, & p \in U \setminus V .
\end{cases}
\]
In particular, $\chi_{\emptyset}(p) = 0$ for every $p \in U$, where $\emptyset$ denotes the empty set.

For \emph{every} $V \subseteq U$, including $V = \emptyset$, the function $\chi_V$ can be extended uniquely to a \emph{valuation}, that is, a function
\[
\overline{\chi_V} : \Gamma \to \{0,1\},
\]
whose value at any formula $\mathcal{A} \in \Gamma$ is determined by the usual truth tables of propositional logic.
A formula $\mathcal{A}$ is said to be a \emph{tautology} if
\[
\overline{\chi_V}(\mathcal{A}) = 1
\quad \text{for every } V \subseteq U .
\]
Using this definition, one verifies, for example, that $(p \vee \neg p)$, $p \to (q \to p)$, and $\neg p \to (p \to q)$ are tautologies.

In this context, the symbols $0$ and $1$ are referred to as \emph{truth values}, commonly read as \emph{false} and \emph{true}, respectively.
This terminology, however, does not by itself explain the \emph{meaning} of truth and falsity.
In empirical science, any interpretation of these notions must ultimately rely on meta-theoretic considerations and laboratory procedures
that relate formal language to observations of the external world.
We begin to explore this relationship in Section~\ref{fuzzy}.

\begin{definition}\label{def:model}
Let $\mathcal T\subseteq\Gamma$, and let $\mathfrak M\subseteq U$; the empty set is allowed.
The expression $\overline{\chi_{\mathfrak M}}(\mathcal T)=1$ means that
$\overline{\chi_{\mathfrak M}}(\mathcal A)=1$ for every $\mathcal A\in\mathcal T$.
In this case, we say that $\mathfrak M$ is a \emph{model} of $\mathcal T$.
\end{definition}
This definition is standard; see, for example, \cite[p.~4]{ChangKeisler}.
Allowing $\mathfrak M=\emptyset$ is essential because $\chi_{\emptyset}$ is one of the ordinary classical valuations introduced above.  For the induced theories $\mathcal T_g$ considered later, every model is nevertheless nonempty: surjectivity of $g$ supplies an initial formula $p$ with $g(p)=1$, and such a $p$ is an axiom.

A set $\mathcal T$ of formul\ae\ is said to be \emph{consistent} if it has a model.
A formula $\mathcal A$ is a \emph{logical consequence} of $\mathcal T$ if every model of $\mathcal T$ is also a model of $\mathcal A$.
The \emph{modus ponens} rule holds: if $A$ and $A \to B$ are logical consequences of $\mathcal T$, then so is $B$.
If $\mathcal T$ is consistent and contains all of its logical consequences, then $\mathcal T$ is called a \emph{theory}.
A theory is often presented in textbooks as the set of all logical consequences of a chosen set of \emph{axioms}, that is, a consistent set of formul\ae\ which, for reasons specific to the given field of study, are taken to be true and sufficiently comprehensive.

\subsection{The Lindenbaum algebra of a theory}
\label{subsec:lind}

Let $\mathcal T$ be a consistent theory.
Two formul\ae\ $\mathcal A$ and $\mathcal B$ are said to be \emph{equivalent modulo $\mathcal T$}, written $\mathcal A \equiv \mathcal B$, if
$(\mathcal A \leftrightarrow \mathcal B) \in \mathcal T$.
The equivalence class of a formula $\mathcal A$ is denoted by $[\mathcal A]$.
The set of all equivalence classes forms a Boolean algebra, denoted by $\mathbb{L}(\mathcal T)$, called the \emph{Lindenbaum algebra} of $\mathcal T$ \cite{Enderton,ChangKeisler}.

The partial order on $\mathbb{L}(\mathcal T)$ is defined by
\[
[\mathcal A] \leq [\mathcal B] \quad \text{if and only if} \quad (\mathcal A \to \mathcal B) \in \mathcal T.
\]
The least upper bound and greatest lower bound of $[\mathcal A]$ and $[\mathcal B]$ are given by
\[
[\mathcal A] \vee [\mathcal B] = [\mathcal A \vee \mathcal B],
\qquad
[\mathcal A] \wedge [\mathcal B] = [\mathcal A \wedge \mathcal B],
\]
respectively.
The zero element is
\[
0 = [\neg \mathcal A \wedge \mathcal A],
\]
and the unit element is
\[
1 = [\neg \mathcal A \vee \mathcal A],
\]
where $\mathcal A$ is any formula.
The Boolean complement of $[\mathcal A]$ is given by
\[
[\mathcal A]' = [\neg \mathcal A].
\]
\section{The core of a theory}
\label{sec:core}

We now introduce the central concept of this paper.
Let $\mathcal T$ be a consistent theory.
For an initial formula $p \in U$, the equivalence class of $p$ is defined by
\begin{equation}
[p] = \{ \mathcal A \in \Gamma : (\mathcal A \leftrightarrow p) \in \mathcal T \}.
\label{eq:core-class}
\end{equation}

The \emph{core} of the theory $\mathcal T$, previously called the \emph{poset} in earlier work \cite{malhas1992quantum}, is the subset
\[
\mathbb C(\mathcal T) \subseteq \mathbb L(\mathcal T)
\]
whose elements are precisely the equivalence classes of initial formul\ae.
Since $\mathbb C(\mathcal T)$ is a subset of the Lindenbaum algebra $\mathbb L(\mathcal T)$, it inherits a natural partial order.

\begin{definition}
\label{def:core-order}
For all $p,q \in U$, we define
\[
[p] \le [q] \quad \text{if and only if} \quad (p \to q) \in \mathcal T.
\]
\end{definition}

In general, very little can be said about the structure of the core of a theory beyond the fact that it is the subset of the Lindenbaum algebra determined by the initial formul\ae.
Any additional properties of $\mathbb C(\mathcal T)$ depend on the specific theory $\mathcal T$.

We shall show that:
\begin{enumerate}
\item[(i)] There exists a consistent theory $\mathcal T$ whose core is not merely a partially ordered set, but also a nontrivial \emph{partial Boolean algebra} (see the next section).
\item[(ii)] Every partial Boolean algebra—such as the partial Boolean algebra of subspaces of a finite-dimensional Hilbert space—is (isomorphic to) the core of a consistent theory formulated within classical propositional logic with its usual logical apparatus.
\end{enumerate}

\section{Partial Boolean algebras}
\label{sec:pba}

We use a component presentation of partial Boolean algebras.
\begin{definition}\label{def:pba}
A
\emph{partial Boolean algebra}, abbreviated $\pba$, is a pair
$(\mathcal V,\Pi)$, where $\mathcal V$ is a nonempty set of vertices and
$\Pi$ is a family of nonempty subsets of $\mathcal V$, called
\emph{components}.  Each component is equipped with a Boolean-algebra
structure, and the following conditions are required; compare
\cite[pp.~126--127]{Varadarajan}.

\begin{enumerate}
\item[\rm(P1)] Every vertex belongs to at least one component.  Two vertices
$x,y$ are called \emph{compatible}, written $x\simeq y$, when they belong to a
common component.

\item[\rm(P2)] Every $\mathcal B\in\Pi$ is a Boolean algebra, and all
components have the same zero and unit, denoted $\mathbf 0$ and $\mathbf 1$.
We assume throughout that $\mathbf 0\ne\mathbf 1$; a partial Boolean algebra
with $\mathbf 0=\mathbf 1$ is called \emph{degenerate} and is excluded.

\item[\rm(P3)] If $\mathcal B_1,\mathcal B_2\in\Pi$, then
$\mathcal B_1\cap\mathcal B_2$ is a component and is a Boolean subalgebra of
both $\mathcal B_1$ and $\mathcal B_2$.  The two Boolean structures agree on
the intersection.

\item[\rm(P4)] If $\mathcal B_1\subseteq\mathcal B_2$ are components, then
$\mathcal B_1$ is a Boolean subalgebra of $\mathcal B_2$.

\item[\rm(P5)] Every finite set of pairwise compatible vertices is contained
in a component.

\item[\rm(P6)] For $x,y\in\mathcal V$, set $x\le y$ when $x$ and $y$ lie in a
common component and $x\le y$ in that Boolean algebra.  This relation is
required to be a partial order on $\mathcal V$.
\end{enumerate}

\end{definition}

Condition~\rm(P3) makes the local Boolean operations independent of the
component in which they are computed.  Thus, whenever $x\simeq y$, the elements
$x\sqcap y$ and $x\sqcup y$ are well defined; and the Boolean complement of
$x$ has the same value in every component containing $x$.  We denote that
complement by $x^\perp$.

\begin{proposition}\label{prop:local-operations}
Let $x,y\in\mathcal V$ be compatible.  If $\mathcal B_1$ and $\mathcal B_2$
are components containing $x$ and $y$, then the meet, join, and complements
of $x$ and $y$ computed in $\mathcal B_1$ agree with those computed in
$\mathcal B_2$.
\end{proposition}

\begin{proof}
The intersection $\mathcal B_1\cap\mathcal B_2$ is, by~\rm(P3), a Boolean
subalgebra of each component with the same induced Boolean structure.  Since
$x,y\in\mathcal B_1\cap\mathcal B_2$, their meet, join, and complements belong
to the intersection and have the same values in both components.
\end{proof}

\begin{definition}\label{def:orth}
For $x,y\in\mathcal V$, we write $x\perp y$ and say that $x$ is
\emph{orthogonal} to $y$ when $x\simeq y$ and $x\le y^\perp$.
\end{definition}

By Proposition~\ref{prop:local-operations}, orthogonality does not depend on
the chosen common component.  It is symmetric: inside any component
containing $x$ and $y$, the Boolean inequality $x\le y^\perp$ is equivalent
to $y\le x^\perp$.

\begin{corollary}\label{cor:weak-oml}
If $a\le b$ in a partial Boolean algebra, then $a$ and $b$ are compatible,
$a^\perp\sqcap b$ is defined, and
\[
 b=a\sqcup(a^\perp\sqcap b).
\]
\end{corollary}

\begin{proof}
By~\rm(P6), $a$ and $b$ lie in a common Boolean component.  The displayed
identity is the ordinary Boolean identity computed in that component, and
Proposition~\ref{prop:local-operations} shows that its value is independent of
the component chosen.
\end{proof}

\begin{remark}\label{rem:global-order}
The transitivity of the global relation $\le$ is part of~\rm(P6); it does not
follow from pairwise compatibility alone.  Condition~\rm(P5) may be invoked
only after all pairs in the finite set have been shown to be compatible.
\end{remark}

\begin{figure}[ht]
\centering
\includegraphics[width=0.35\textwidth]{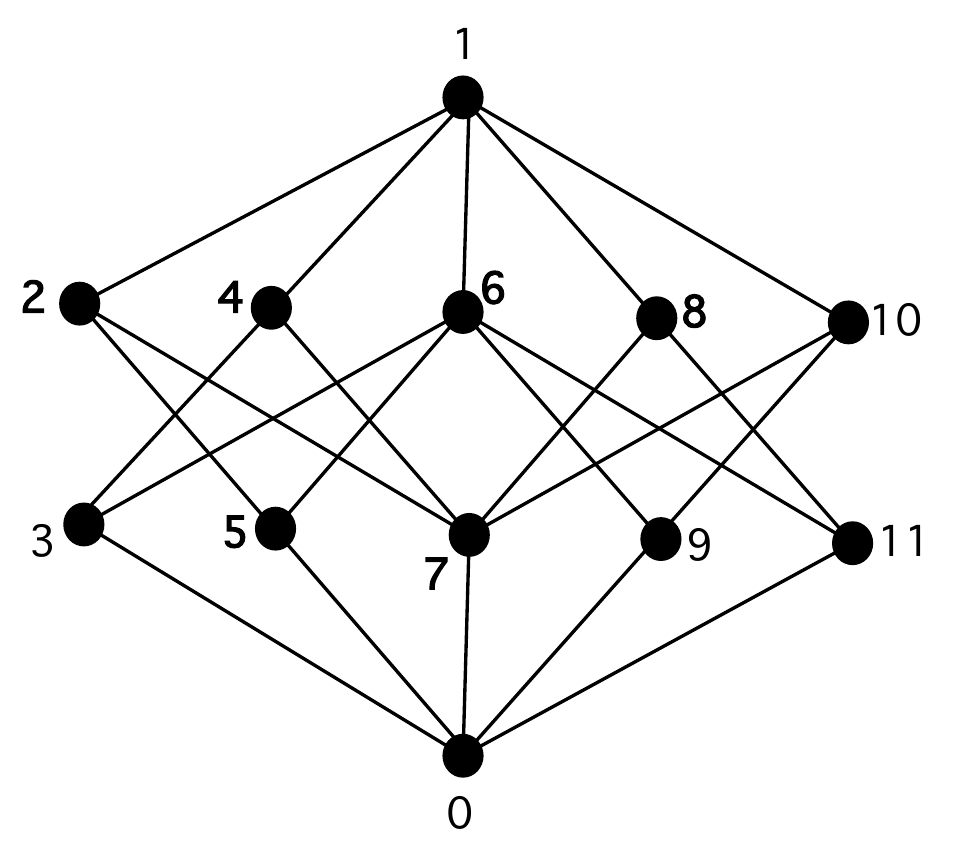}\hfill
\includegraphics[width=0.55\textwidth]{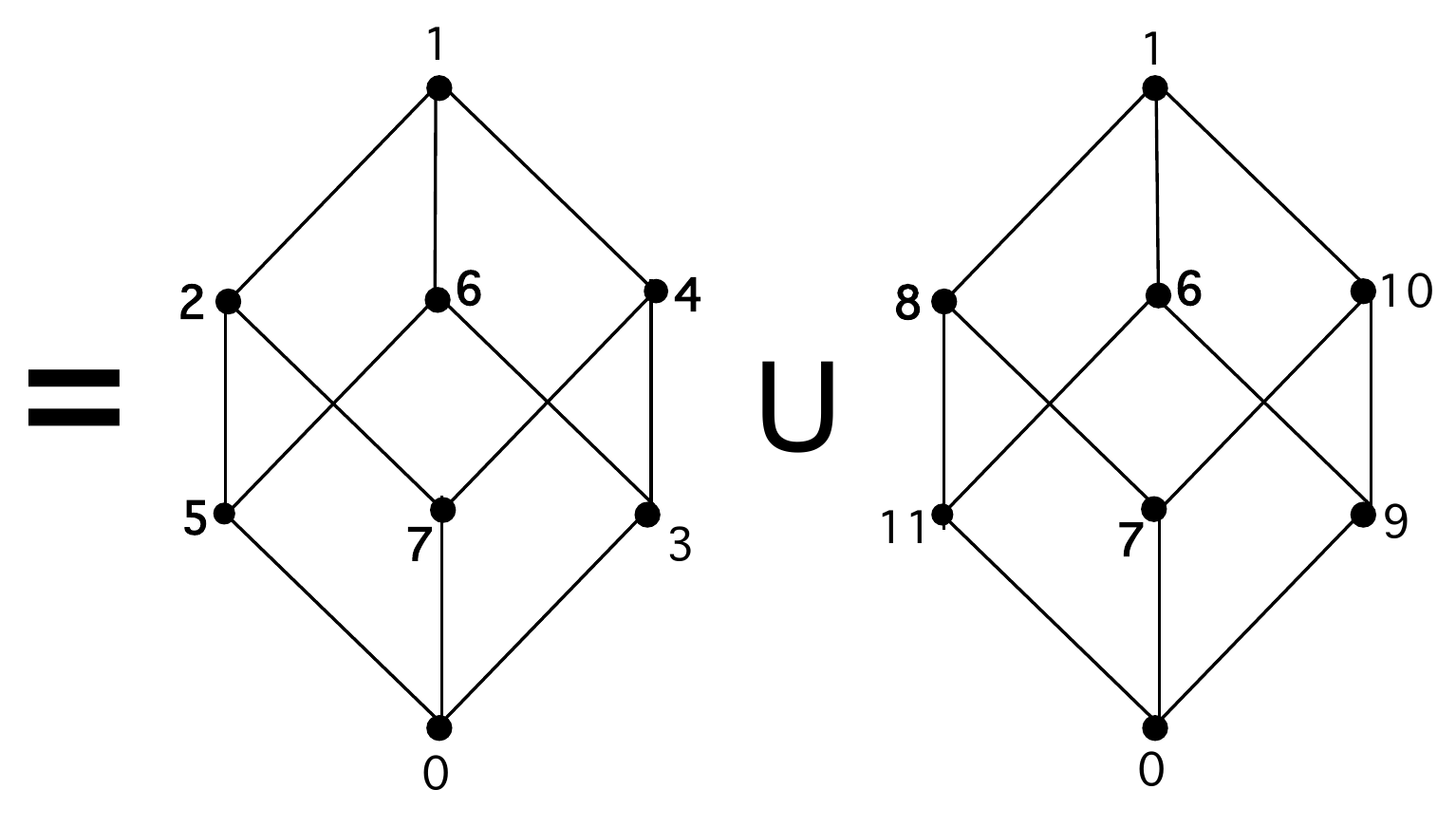}

\vspace{0.5em}

$\hspace{7em}\mathcal V_{12} \hfill \hspace{7em}\mathcal B_1 \hfill \mathcal B_2\hspace{4em}$

\caption{\small
The Hasse diagram $\mathcal V_{12}$, shown on the left, is an
orthocomplemented lattice obtained by gluing the Boolean algebras
$\mathcal B_1$ and $\mathcal B_2$ along their common Boolean subalgebra
$\{0,1,6,7\}$.  With the component family specified in
\eqref{eq:p12}, it is also a partial Boolean algebra.
}
\label{fig:v12-pba}
\end{figure}

\section{A simple nontrivial example of a partial Boolean algebra}
\label{sec:simple-pba}

For a finite-dimensional Hilbert space, the quantum propositions may be
represented by its linear subspaces.  Compatible propositions are those whose
orthogonal projections commute; every finite compatible family lies in a
Boolean algebra generated by a common orthogonal decomposition.  In this way
the subspaces form a partial Boolean algebra \cite{Varadarajan}.

A smaller finite example is the set $\mathcal V_{12}$ of the twelve vertices
in Figure~\ref{fig:v12-pba}.  The vertices are labelled
$0,1,\ldots,11$.  As a lattice, $\mathcal V_{12}$ is
orthocomplemented and nondistributive \cite{ZierlerSchlessinger}.
Orthocomplementation is given by
\begin{equation}\label{eq:oc}
1=0^\perp,\quad 3=2^\perp,\quad 5=4^\perp,\quad
7=6^\perp,\quad 9=8^\perp,\quad 11=10^\perp,
\qquad (n^\perp)^\perp=n.
\end{equation}

Let $\mathcal B_1$ and $\mathcal B_2$ be the two eight-element Boolean
algebras shown in Figure~\ref{winnie}.  Their intersection is the
four-element Boolean algebra $\{0,1,6,7\}$.  Define
\begin{equation}\label{eq:p12}
\Pi^\circ=
\bigl\{
\{0,1\},
\{0,6,7,1\},
\{0,4,5,1\},
\{0,2,3,1\},
\{0,8,9,1\},
\{0,10,11,1\},
\mathcal B_1,
\mathcal B_2
\bigr\}.
\end{equation}
These are exactly the Boolean subalgebras of the two maximal components.

\begin{proposition}\label{prop:v12-pba}
The pair $(\mathcal V_{12},\Pi^\circ)$ is a partial Boolean algebra.  Its
maximal components are $\mathcal B_1$ and $\mathcal B_2$, and two vertices
are compatible if and only if they lie together in one of these two maximal
components.
\end{proposition}

\begin{proof}
Every vertex lies in $\mathcal B_1$ or $\mathcal B_2$, and both maximal
components share the same $0$, $1$, and orthocomplementation.  Their
intersection $\{0,1,6,7\}$ is a Boolean subalgebra of each.  The family
$\Pi^\circ$ contains this intersection and every Boolean subalgebra of either
maximal component, so~\rm(P1)--\rm(P4) hold.

If a finite family is pairwise compatible and contains a vertex from
$\{2,3,4,5\}$, then it contains no vertex from $\{8,9,10,11\}$ and is
contained in $\mathcal B_1$; the symmetric argument places any family meeting
$\{8,9,10,11\}$ in $\mathcal B_2$.  A family meeting neither set is contained
in their intersection.  Hence~\rm(P5) holds.

Finally, the order induced by either maximal component agrees with the order
shown in the Hasse diagram of $\mathcal V_{12}$.  It is therefore the
restriction of a single lattice order and is a partial order, proving~\rm(P6).
The characterization of compatibility and maximality is immediate.
\end{proof}

Consequently,
\[
x\not\simeq y
\quad\Longleftrightarrow\quad
\bigl(x\in\{2,3,4,5\}\ \text{and}\ y\in\{8,9,10,11\}\bigr)
\]
up to interchanging $x$ and $y$.

\begin{figure}[ht]
\centering
\includegraphics[width=0.56\textwidth]{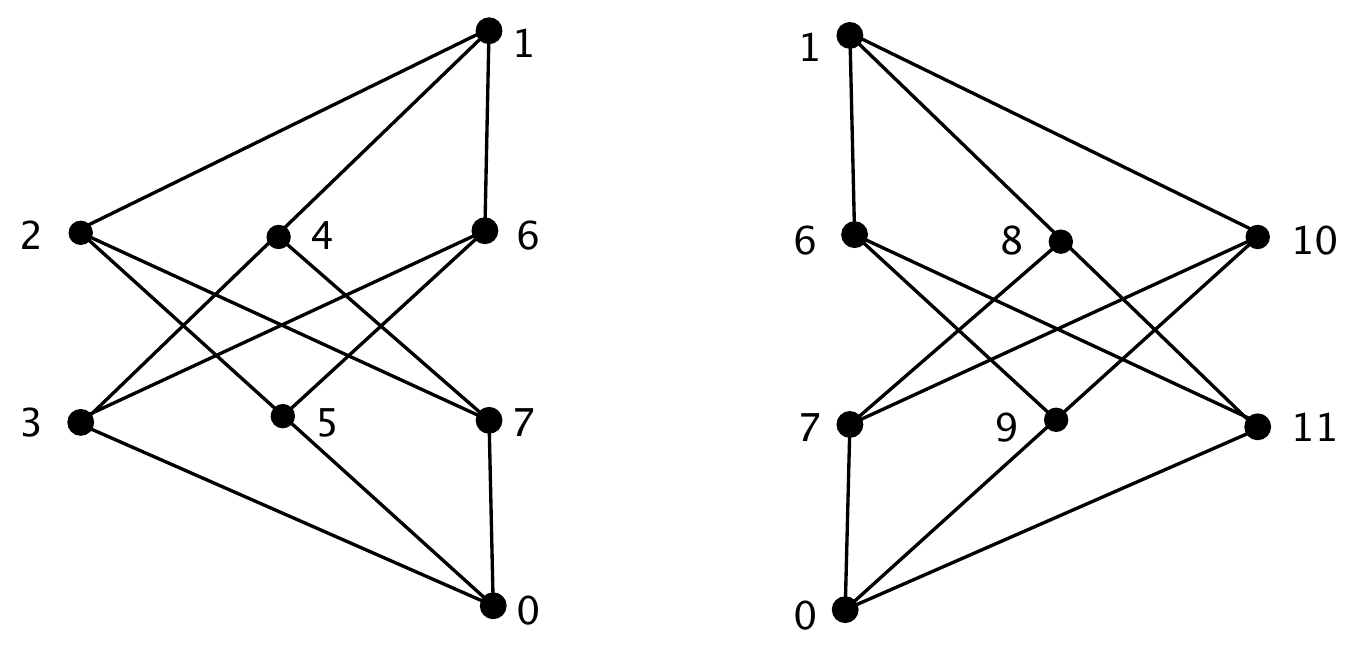}

\vspace{0.5em}

$\hspace{18em}\mathcal B_1 \hfill \mathcal B_2\hspace{18em}$

\caption{\small
The two maximal Boolean components of $\mathcal V_{12}$.  They are glued
along the common Boolean subalgebra $\{0,1,6,7\}$; orthocomplementation is
given by~\eqref{eq:oc}.
}
\label{winnie}
\end{figure}

\begin{definition}\label{def:pba-iso}
An isomorphism
$h:(\mathcal V_1,\Pi_1)\to(\mathcal V_2,\Pi_2)$ of partial Boolean algebras
is a bijection $h:\mathcal V_1\to\mathcal V_2$ such that
\[
\mathcal B\in\Pi_1
\quad\Longleftrightarrow\quad
h(\mathcal B)\in\Pi_2,
\]
and, for every component $\mathcal B\in\Pi_1$, the restriction
$h|_{\mathcal B}$ is a Boolean-algebra isomorphism onto $h(\mathcal B)$.
\end{definition}
\section{The theory induced by $g:U\twoheadrightarrow \mathcal V$}
\label{sec:gaxioms}

Let $(\mathcal V,\Pi)$ be an arbitrary $\pba$, and suppose that $U$, the set of initial formul\ae\ of a propositional language, is sufficiently large to admit a surjection
\[
g:U\twoheadrightarrow \mathcal V .
\]
Intuitively, the function $g$ \emph{sprinkles} the set $U$ over $\mathcal V$ so that every vertex
$x\in\mathcal V$ has at least one initial formula of $U$ assigned to it.

\begin{definition}
\label{def:compat-initial}
Initial formul\ae\ $p,q\in U$ are said to be \emph{compatible}, written $p\simeq q$, if and only if there exists a component of $(\mathcal V,\Pi)$ containing both $g(p)$ and $g(q)$.
Equivalently, $p\simeq q$ if and only if $g(p)\simeq g(q)$.
\end{definition}

Thus, initial formul\ae\ $p,q$ are compatible if and only if the vertices $g(p)$ and $g(q)$ are compatible.
In particular, if $g(p)=\mathbf 1$ and $g(q)=\mathbf 0$, then $p\simeq q$, since the vertices $\mathbf 0$ and $\mathbf 1$ belong to every component.
It also follows directly from the definitions that if $g(p)$ is orthogonal to $g(q)$, then $g(p)\le g(q)^\perp$ and hence $p\simeq q$;
see condition~\rm(P6) of Definition~\ref{def:pba} and Definition~\ref{def:orth}.

Let $\mathbf 0$ denote the lowest (zero) element of $\mathcal V$ and $\mathbf 1$ the highest (unit) element.
The fact that a $\pba$ is partially ordered and orthocomplemented allows us to specify a collection of formul\ae\ that serve as the axioms of a \emph{consistent} theory $\mathcal T_g$ whose core is a $\pba$ isomorphic to $(\mathcal V,\Pi)$.

\subsection{Axioms of $\mathcal T_g$}
\label{sec:mal2}

Let $(\mathcal V,\Pi)$ be a $\pba$, and let $U$ be the set of all initial formul\ae\ of a propositional language.
Assume that $U$ is large enough to admit a surjection $g:U\twoheadrightarrow\mathcal V$.
Each such surjection determines a consistent set of \emph{axioms} as follows.

A formula $\mathcal A$ is called an \emph{axiom induced by $g$}, or simply an \emph{axiom}, if and only if one of the following four conditions is satisfied \cite{mal2}:
\begin{enumerate}
\item[(g1)] $\mathcal A=p$, where $p\in U$ and $g(p)=1$;
\item[(g2)] $\mathcal A=\neg p$, where $p\in U$ and $g(p)=0$;
\item[(g3)] $\mathcal A=(p\to q)$, where $p,q\in U$ and $g(p)\le g(q)$;
\item[(g4)] $\mathcal A=(p\to\neg q)$, where $p,q\in U$ and $g(p)\perp g(q)$.
\end{enumerate}

The compatibility condition $p\simeq q$ in items (g3) and (g4) is superfluous, since both $g(p)\le g(q)$ and $g(p)\perp g(q)$ (that is, $g(p)\le g(q)^\perp$) already presuppose compatibility; see condition~\rm(P6) in Definition~\ref{def:pba} and Definition~\ref{def:orth}.

The theory generated by these axioms is called the \emph{theory induced by $g$} and is denoted by $\mathcal T_g$.
We refer to the surjection $g:U\to\mathcal V$ as the \emph{canonical function}, since it determines a consistent set of laws, axioms, or \emph{canons} governing a fragment of external reality.
We shall show below that $\mathcal T_g$ is consistent.

\subsection{Clusters}
\label{sec:clusters}

Elements $p,q\in U$ are said to be \emph{mutually orthogonal}, written $p\perp q$, if and only if
$g(p)\perp g(q)$, that is, $g(p)\simeq g(q)$ and $g(p)\le g(q)^\perp$.

\begin{definition}
\label{def:cluster}
A nonempty subset $C\subseteq U$ is called a \emph{cluster} if it contains no
pair of mutually orthogonal elements, where the two members of the pair are
\emph{not required to be distinct}.  Explicitly: if $p,q\in C$, then
$g(p)\not\perp g(q)$; in particular no $p\in C$ satisfies $g(p)\perp g(p)$,
that is, $g(p)\ne\mathbf 0$ for every $p\in C$.
\end{definition}

The parenthetical clause is not a technicality.  Reading ``pair'' as ``pair of
distinct elements'' would admit sets containing a formula mapped to
$\mathbf 0$, and Lemma~\ref{lem:zero} and Theorem~\ref{thmC} below would both
fail.  The same convention is used in Lemma~\ref{lem:notmodel}, where $p=q$ is
allowed.

For example, let $U_{12}$ be a collection of twelve initial formul\ae\
\[
U_{12}=\{p_0,p_1,p_2,p_3,p_4,p_5,p_6,p_7,p_8,p_9,p_{10},p_{11}\},
\]
and let $g:U_{12}\twoheadrightarrow\mathcal V_{12}$ be defined by
\begin{equation}
g(p_n)=n,
\label{spkl1}
\end{equation}
so that the initial formula $p_n$ is mapped to the vertex labeled $n\in\{0,1,2,\dots,11\}$.
With Figure~\ref{fig:v12-pba} in view, each of the following sets is a cluster:
\[
\{p_2,p_4,p_6\},\quad
\{p_5,p_{11}\},\quad
\{p_5,p_{11},p_8\},\quad
\{p_5,p_9,p_6\}.
\]
By contrast, the sets
\[
\{p_5,p_6,p_4\},\;
\{p_5,p_7,p_3\},\;
\{p_5,p_{11},p_{10}\},\;
\{p_0,p_5\}
\]
are not clusters.
\subsection{Models of $\mathcal T_g$}
\label{modA}

Recall (see Section~\ref{sec:basics}) that if $V \subseteq U$, then
$\chi_V : U \to \{0,1\}$ denotes the characteristic function of $V$ and
$\overline{\chi_V}$ its unique extension to a valuation on $\Gamma$.

\begin{lemma}
\label{lem:notmodel}
If $V \subseteq U$, $p,q \in V$, and $g(p) \perp g(q)$, then $V$ is \emph{not}
a model of\/ $\mathcal T_g$.
\end{lemma}

\begin{proof}
Assume $V \subseteq U$, $p,q \in V$, and $g(p) \perp g(q)$.
Then $p \to \neg q$ is an axiom of $\mathcal T_g$.
If $V$ were a model, we would have
$\overline{\chi_V}(p \to \neg q)=1$.
Since $p \in V$, $\overline{\chi_V}(p)=1$, hence
$\overline{\chi_V}(\neg q)=1$, so $\overline{\chi_V}(q)=0$.
But $q$ is an initial formula and $q \in V$, contradicting
$\chi_V(q)=1$.
\end{proof}

\begin{lemma}
\label{lem:zero}
If $C \subseteq U$ is a cluster and $g(p)=\mathbf 0$, then $p \notin C$.
\end{lemma}

\begin{proof}
The element $\mathbf 0$ is orthogonal to every vertex of $\mathcal V$, and in
particular $\mathbf 0\perp\mathbf 0$.  If $p$ belonged to $C$, then
$g(p)\perp g(p)$, which is excluded by Definition~\ref{def:cluster}.  Hence
$p\notin C$.
\end{proof}

We now show that a partial converse of Lemma~\ref{lem:notmodel} holds:
every cluster can be extended to a model of the axioms.

\begin{lemma}
\label{lem:clusterstar}
Let $C \subseteq U$ be a cluster and define
\begin{equation}
\label{eq:clusterstar}
C^* = \{\, q \in U : g(r)\le g(q) \text{ for some } r \in C \,\}.
\end{equation}
Then $C^*$ is also a cluster.
\end{lemma}

\begin{proof}
Suppose $p_0,q_0 \in C^*$ and $g(p_0)\perp g(q_0)$.
Then there exist $r_1,r_2 \in C$ such that
$g(r_1)\le g(p_0)$ and $g(r_2)\le g(q_0)$.
Since $g(p_0)\perp g(q_0)$, we have $g(p_0)\le g(q_0)^\perp$, hence
\[
g(r_1)\le g(p_0)\le g(q_0)^\perp \le g(r_2)^\perp,
\]
so $g(r_1)\perp g(r_2)$, contradicting that $C$ is a cluster.
\end{proof}

We next show that $C^*$ is a model of the axioms.

\begin{theorem}
\label{thmC}
Let $C\subseteq U$ be a cluster and let $C^*$ be defined
by~\eqref{eq:clusterstar}.
Then $\overline{\chi_{C^*}}(\mathcal A)=1$ for every axiom $\mathcal A$ of
$\mathcal T_g$.
Moreover, clusters exist, so $\mathcal T_g$ is consistent.
\end{theorem}

\begin{proof}
By Section~\ref{sec:mal2}, there are four types of axioms.

\begin{enumerate}
\item If $\mathcal A=p$ and $g(p)=1$, then $g(r)\le g(p)$ for all $r\in C$,
so $p\in C^*$ and $\overline{\chi_{C^*}}(p)=1$.

\item If $\mathcal A=\neg p$ and $g(p)=0$, then no $r\in C$ satisfies
$g(r)\le g(p)=\mathbf 0$, since by Lemma~\ref{lem:zero} no element of the cluster
$C$ is mapped to $\mathbf 0$.
Hence $p\notin C^*$ and $\overline{\chi_{C^*}}(\neg p)=1$.

\item If $\mathcal A=p\to q$ with $g(p)\le g(q)$ and
$\overline{\chi_{C^*}}(p\to q)=0$, then
$\overline{\chi_{C^*}}(p)=1$ and $\overline{\chi_{C^*}}(q)=0$.
Thus $p\in C^*$, so $g(r)\le g(p)\le g(q)$ for some $r\in C$,
implying $q\in C^*$, a contradiction.

\item If $\mathcal A=p\to\neg q$ with $g(p)\perp g(q)$ and
$\overline{\chi_{C^*}}(p\to\neg q)=0$, then
$p,q\in C^*$, so there exist $r_1,r_2\in C$ with
$g(r_1)\le g(p)$ and $g(r_2)\le g(q)$.
This implies $g(r_1)\perp g(r_2)$, contradicting that $C$ is a cluster.
\end{enumerate}

Finally, a cluster exists.  Since the partial Boolean algebra is nondegenerate
by~\rm(P2), we have $\mathbf 1\ne\mathbf 0$, so $\mathbf 1\not\perp\mathbf 1$;
choosing $p$ with $g(p)=\mathbf 1$, which is possible because $g$ is onto, the
singleton $\{p\}$ is a cluster.  Hence $\{p\}^*$ is a model and
$\mathcal T_g$ is consistent.
\end{proof}
We call $C^*$ the \emph{model based on the cluster $C$}.
Nondegeneracy is genuinely needed here: if $\mathbf 0=\mathbf 1$, then a single
initial formula $p$ satisfies both~(g1) and~(g2), so $p$ and $\neg p$ are both
axioms and $\mathcal T_g$ is inconsistent.

\begin{theorem}
\label{thm:models-are-clusters}
If $\mathfrak M$ is a model of $\mathcal T_g$, then $\mathfrak M$ is a cluster
and satisfies $\mathfrak M=\mathfrak M^*$.
\end{theorem}

\begin{proof}
By Lemma~\ref{lem:notmodel}, $\mathfrak M$ is a cluster; it is nonempty because
every $p$ with $g(p)=\mathbf 1$ is an axiom.
If $p\in\mathfrak M$ and $g(p)\le g(q)$, then $p\to q$ is an axiom of
$\mathcal T_g$.  Because $\mathfrak M$ is a model,
$\overline{\chi_{\mathfrak M}}(p\to q)=1$, and
$\chi_{\mathfrak M}(p)=1$; hence
$\chi_{\mathfrak M}(q)=1$, that is, $q\in\mathfrak M$.
Therefore $\mathfrak M^*\subseteq\mathfrak M$.  The reverse inclusion is
immediate from reflexivity of $\le$: for every $q\in\mathfrak M$,
$g(q)\le g(q)$, so $q\in\mathfrak M^*$.  Thus
$\mathfrak M=\mathfrak M^*$.
\end{proof}

\begin{corollary}
\label{negifforth}
For all $p,q\in U$,
\[
p\to\neg q \in \mathcal T_g \quad \text{if and only if} \quad g(p)\perp g(q).
\]
\end{corollary}
\begin{proof}
If $g(p)\perp g(q)$, then $p\to\neg q$ is an axiom and hence belongs to
$\mathcal T_g$.

Conversely, suppose $p\to\neg q\in\mathcal T_g$ but $g(p)\not\perp g(q)$.  In
particular $g(p)\ne\mathbf 0$ and $g(q)\ne\mathbf 0$, since $\mathbf 0$ is
orthogonal to everything.  Hence $C=\{p,q\}$ contains no orthogonal pair, so it
is a cluster, and by Theorem~\ref{thmC} the set $C^*$ is a model of
$\mathcal T_g$.  Both $p$ and $q$ lie in $C^*$, so
$\chi_{C^*}(p)=\chi_{C^*}(q)=1$ and therefore
$\overline{\chi_{C^*}}(p\to\neg q)=0$.  This contradicts
$p\to\neg q\in\mathcal T_g$.
\end{proof}
\begin{theorem}
\label{thm:equiv}
For all $p,q\in U$,
\[
p\to q \in \mathcal T_g \quad \text{if and only if} \quad g(p)\le g(q).
\]
Consequently, $[p]=[q]$ if and only if $g(p)=g(q)$.
\end{theorem}

\begin{proof}
If $g(p)\le g(q)$, then $p\to q$ is an axiom of $\mathcal T_g$, so
$p\to q\in\mathcal T_g$.

Conversely, suppose $p\to q\in\mathcal T_g$ and $g(p)\not\le g(q)$.
The latter inequality implies $g(p)\ne\mathbf 0$, because
$\mathbf 0\le x$ for every $x\in\mathcal V$.  Hence $\{p\}$ is a cluster,
and Theorem~\ref{thmC} shows that
\[
\{p\}^*=\{r\in U:g(p)\le g(r)\}
\]
is a model of $\mathcal T_g$.  We have $p\in\{p\}^*$ but
$q\notin\{p\}^*$, so this model assigns value $0$ to $p\to q$, contradicting
$p\to q\in\mathcal T_g$.  Therefore $g(p)\le g(q)$.

Finally, $[p]=[q]$ if and only if both $p\to q$ and $q\to p$ belong to
$\mathcal T_g$.  By the first part, this is equivalent to
$g(p)\le g(q)$ and $g(q)\le g(p)$, hence to $g(p)=g(q)$.
\end{proof}

\begin{corollary}
\label{cor:mp}
If $C\subseteq U$ is a cluster, then
\begin{equation}
\label{qluster}
C^*=\{\,q\in U:(r\to q)\in\mathcal T_g
          \text{ for some }r\in C\,\}.
\end{equation}
Thus $C^*$ is precisely the set of initial formul\ae\ obtainable from members
of $C$ by implications belonging to $\mathcal T_g$.
\end{corollary}

\begin{proof}
By definition, $q\in C^*$ if and only if $g(r)\le g(q)$ for some $r\in C$.
Theorem~\ref{thm:equiv} identifies this condition with
$(r\to q)\in\mathcal T_g$.
\end{proof}

\begin{corollary}
\label{cor:star-properties}
Let $C,C_1,C_2\subseteq U$ be clusters.
\begin{enumerate}
\item $C^*$ is a model and $(C^*)^*=C^*$.
\item If $C_1\subseteq C_2$, then $C_1^*\subseteq C_2^*$.
\item The inclusion in item~(2) is strict if and only if there exists
      $q\in C_2$ such that $q\notin C_1^*$.
\end{enumerate}
\end{corollary}

\begin{proof}
The first assertion follows from Theorems~\ref{thmC} and
\ref{thm:models-are-clusters}.  The second follows directly from the definition
of $C^*$.  For the third, if some $q\in C_2\setminus C_1^*$, then
$q\in C_2^*$ by reflexivity, so the inclusion is strict.  Conversely, if
$C_2\subseteq C_1^*$, then monotonicity and idempotence give
$C_2^*\subseteq(C_1^*)^*=C_1^*$; together with item~(2), this yields equality.
\end{proof}
\begin{theorem}
\label{thm:corepba}
The core $\mathbb C(\mathcal T_g)$ of\/ $\mathcal T_g$ carries a unique partial
Boolean algebra structure with respect to which the map
$f([p])=g(p)$ is an isomorphism onto $(\mathcal V,\Pi)$.
\end{theorem}

\begin{proof}
Define $f:\mathbb C(\mathcal T_g)\to\mathcal V$ by $f([p])=g(p)$.

\emph{$f$ is a well-defined bijection.}
By Theorem~\ref{thm:equiv}, $[p]=[q]$ if and only if $g(p)=g(q)$; hence $f$ is
well defined and injective.
Since $g$ is surjective, so is $f$.

\emph{$f$ is an order isomorphism.}
The order on $\mathbb C(\mathcal T_g)$ is the one inherited from the Lindenbaum
algebra, namely $[p]\le[q]$ iff $(p\to q)\in\mathcal T_g$.
By Theorem~\ref{thm:equiv}, $(p\to q)\in\mathcal T_g$ iff $g(p)\le g(q)$, so
$[p]\le[q]$ in $\mathbb C(\mathcal T_g)$ if and only if $f([p])\le f([q])$ in
$\mathcal V$.

\emph{Transport of structure.}
Declare two elements $[p],[q]$ of the core to be \emph{compatible} iff
$g(p)\simeq g(q)$, and take the components of the core to be the sets
$f^{-1}(\mathcal B)=\{[p]:g(p)\in\mathcal B\}$ for $\mathcal B\in\Pi$;
for each $[p]$ set $[p]^\perp=f^{-1}\!\bigl(g(p)^\perp\bigr)$, which is
well defined by injectivity of $f$.
With these definitions $f$ maps components bijectively onto components and, being
an order isomorphism, restricts on each component $f^{-1}(\mathcal B)$ to a Boolean
isomorphism onto $\mathcal B$ that preserves complementation.
Thus $(\mathbb C(\mathcal T_g),f^{-1}(\Pi))$ is a partial Boolean algebra and $f$
is an isomorphism onto $(\mathcal V,\Pi)$ in the sense of
Definition~\ref{def:pba-iso}.
Uniqueness is immediate: any partial Boolean algebra structure on the core making
$f$ an isomorphism must have exactly these components, complement, and order.
\end{proof}

\section{The Lindenbaum algebra and the core of $\mathcal T_g$}
\label{sec:compare}

The Lindenbaum algebra of every consistent theory formulated in classical
propositional logic is a Boolean algebra.
This fact is so fundamental that it is often regarded as a defining
property of classical logic itself.
Thus, for example, at the outset of their seminal paper,
Kochen and Specker \cite{KS} write:
\emph{``The classical propositional calculus is essentially Boolean algebra.''}

In a similar spirit, and as a generalization of this observation,
\emph{quantum logic} is often identified with the partial Boolean algebra
of \emph{quantum propositions}, namely the partial Boolean algebra of closed
subspaces of the Hilbert space associated with a quantum system.
In what follows, we show that quantum logic, understood in this sense,
\emph{arises as the core of a consistent theory formulated entirely within
the language of classical propositional calculus, equipped with its usual
logical apparatus}.

\medskip

For an arbitrary theory $\mathcal T$, very little can be said in general
about its core $\mathbb C(\mathcal T)$ beyond the fact that it is a subset of
the Lindenbaum algebra $\mathbb L(\mathcal T)$:
\begin{equation}
\label{CT1}
\mathbb C(\mathcal T)\subseteq \mathbb L(\mathcal T).
\end{equation}
In particular, $\mathbb C(\mathcal T)$ need not inherit the full Boolean
structure of $\mathbb L(\mathcal T)$.

By contrast, for the theory $\mathcal T_g$ constructed in the previous
sections, the situation is markedly different.
As shown in Theorem~\ref{thm:corepba}, the core $\mathbb C(\mathcal T_g)$ is
not merely a subset of $\mathbb L(\mathcal T_g)$, but is in fact a
\emph{nontrivial partial Boolean algebra}.

\medskip

To clarify the distinction, let $p,q$ be compatible initial formul\ae.
Then the pair $\{[p],[q]\}$ always has a least upper bound in the Lindenbaum
algebra $\mathbb L(\mathcal T_g)$, namely
\[
[p]\vee[q]=[p\vee q].
\]
However, since $p$ and $q$ are compatible, the pair $\{[p],[q]\}$ also has
a least upper bound \emph{within the core} $\mathbb C(\mathcal T_g)$.
This least upper bound is the equivalence class of an initial formula $x$,
and we denote it by
\[
[p]\sqcup[q]=[x].
\]
By~\eqref{CT1} the element $[p]\sqcup[q]$ also lies in
$\mathbb L(\mathcal T_g)$, where it is still an upper bound of
$\{[p],[q]\}$, though in general no longer the least one.  Since $[p]\vee[q]$
\emph{is} the least upper bound in $\mathbb L(\mathcal T_g)$, we obtain
\[
[p]\vee[q]\le [p]\sqcup[q],
\]
and the inequality may be strict.
Similarly, if $\{[p],[q]\}$ has a greatest lower bound in the core, we denote
it by
\[
[p]\sqcap[q]=[y],
\]
where $y$ is an initial formula.
The same argument, applied to lower bounds, gives
\[
[p]\sqcap[q]\le [p]\wedge[q].
\]

\medskip

These observations motivate the following question:
\emph{Does there exist a theory, formulated in the language of classical
propositional calculus, whose core is a given partial Boolean algebra?}
In particular, can the partial Boolean algebra
$(\mathcal V_{12},\Pi^\circ)$ illustrated in Figure~\ref{fig:v12-pba}, or the partial
Boolean algebra of closed subspaces of a Hilbert space, arise as the core of
such a theory?

In the latter case, the resulting core would coincide with what is
commonly referred to as \emph{quantum logic}.
Theorem~\ref{thm:corepba} answers this question affirmatively:
every partial Boolean algebra arises, up to isomorphism, as the core of a
consistent classical propositional theory of the form $\mathcal T_g$.

\subsection{An example: a theory whose core is a nontrivial $\pba$}
\label{orand}

With Figure~\ref{fig:v12-pba} and the canonical function $g$ defined in
\eqref{spkl1} in mind, let
\[
C=\{p_5,p_8\}.
\]
Then $C$ is a cluster, and its extension
\[
C^*=\{p_5,p_2,p_6,p_8,p_1\}
\]
is a model of the axioms of $\mathcal T_g$.
The initial formul\ae\ $p_9$ and $p_{11}$ do not belong to $C^*$ and are
therefore false in this model.

However, $p_6$ is true in this model, even though
\[
g(p_9)\sqcup g(p_{11})=g(p_6),
\qquad\text{equivalently}\qquad
[p_9]\sqcup[p_{11}]=[p_6].
\]
Thus the operation $\sqcup$ in the core does not behave like the classical
logical connective ``or''.
More generally, it follows from~\eqref{CT1} that whenever the least upper
bound of two elements $[p],[q]\in\mathbb C(\mathcal T_g)$ exists, it is
merely an upper bound of $\{[p],[q]\}$ in the Lindenbaum algebra
$\mathbb L(\mathcal T_g)$.
Similarly, whenever the greatest lower bound exists in the core, it is
merely a lower bound in $\mathbb L(\mathcal T_g)$.

This phenomenon is also reflected in the behavior of meets.
For example, the initial formul\ae\ $p_2$ and $p_8$ are both true in the
model $C^*$, yet
\[
g(p_2)\sqcap g(p_8)=g(p_7),
\]
and $p_7\notin C^*$, so $p_7$ is false in this model.

These observations are special cases of Theorem~\ref{thm:corepba}, which identifies
the core of $\mathcal T_g$ with the partial Boolean algebra $(\mathcal V,\Pi)$.
In particular, when $U_{12}=\{p_0,\dots,p_{11}\}$ and
$g(p_n)=n$ as in~\eqref{spkl1}, the core of $\mathcal T_g$ is
isomorphic to the partial Boolean algebra $\mathcal V_{12}$ shown in
Figure~\ref{fig:v12-pba}.

\subsection{Negation and orthocomplementation}
\label{xnot}

The unary operation of orthocomplementation in the core of a theory does not
coincide with negation in the Lindenbaum algebra.
To see this, let $(\mathcal V,\Pi)$ be a partial Boolean algebra and let
$\mathcal T_g$ be the theory induced by a surjection
$g:U\to\mathcal V$.
Suppose that $g(q)=g(p)^\perp$.
Then $g(q)\le g(p)^\perp$, and hence $(q\to\neg p)\in\mathcal T_g$.
Thus the initial formul\ae\ $p$ and $q$ cannot both be true in any model of
$\mathcal T_g$.

One might expect that $p$ and $q$ also cannot both be false.
This, however, is not the case.

\begin{theorem}
\label{notneg}
Orthocomplementation in the core is not negation.
\end{theorem}

\begin{proof}
With Figure~\ref{fig:v12-pba} in mind, let
$U_{12}=\{p_0,p_1,\dots,p_{11}\}$ and let $g:U_{12}\to\mathcal V_{12}$ be
given by~\eqref{spkl1}.
Let $C=\{p_5\}$.
Then $C$ is a cluster, and
\begin{equation}
\label{eq:negexample}
C^*=\{q\in U_{12}: g(p_5)\le g(q)\}
     =\{p_5,p_2,p_6,p_1\}
\end{equation}
is the model of $\mathcal T_g$ based on $C$.
In this model, both $p_8$ and $p_9$ are false, even though
$g(p_8)=g(p_9)^\perp$.
Thus $p_8$ is not logically equivalent to $\neg p_9$.
\end{proof}

Nevertheless, orthocomplementation and negation are closely related.
Since $\mathbb C(\mathcal T_g)\subseteq\mathbb L(\mathcal T_g)$, the order
relation $\le$ of the Lindenbaum algebra applies to elements of the core.
In particular, expressions of the form
$[p]^\perp\le[\neg p]$ are meaningful.

The next theorem shows that $[p]^\perp$ is the largest element of the core
lying below $[\neg p]$ in the Lindenbaum algebra.

\begin{theorem}
\label{negorth1}
$(i)$ For all $p\in U$, $[p]^\perp\le[\neg p]$. 
$(ii)$ If $[r]\le[\neg p]$, then $[r]\le[p]^\perp$.
\end{theorem}

\begin{proof}
(i) Let $[q]=[p]^\perp$.
Then $g(q)=g(p)^\perp$, so $q\to\neg p$ is an axiom of $\mathcal T_g$.
Hence $[q]\le[\neg p]$.

\smallskip

(ii) Suppose $[r]\le[\neg p]$.
Then $(r\to\neg p)\in\mathcal T_g$.
By Corollary~\ref{negifforth}, $g(r)\perp g(p)$, hence
$g(r)\le g(p)^\perp$.
Equivalently, $[r]\le[p]^\perp$.
\end{proof}
The distinction between orthocomplementation in the core and negation in the Lindenbaum algebra reveals that quantum-like logical constraints can emerge entirely within classical propositional logic, not from a failure of classical logic, but from the structural limitations imposed on which propositions may be jointly meaningful.
\section{Finite-dimensional partial Boolean algebras}
\label{sec:finite-pba}

Let $(\mathcal V,\Pi)$ be a partial Boolean algebra.  An element
$\alpha\in\mathcal V$ is an \emph{atom} if $\alpha\ne0$ and
\[
0\le x\le\alpha
\quad\Longrightarrow\quad
x=0\ \text{or}\ x=\alpha.
\]
The partial Boolean algebra is \emph{atomic} if every nonzero vertex lies
above an atom.  Notice that ``lies above an atom'' is weaker than ``covers an
atom''; no covering hypothesis is needed here.

A \emph{frame} is a maximal set of pairwise orthogonal atoms.  By~\rm(P5),
every finite set of pairwise orthogonal atoms lies in a common Boolean
component.

\begin{definition}\label{def:finite-dimensional}
An atomic partial Boolean algebra $(\mathcal V,\Pi)$ is
\emph{$n$-dimensional} if:
\begin{enumerate}
\item every set of pairwise orthogonal atoms is contained in a frame;
\item every frame has exactly $n$ atoms;
\item for every frame $\mathcal E$, there is a unique maximal component
$\mathcal B_{\mathcal E}$ whose set of Boolean atoms is precisely
$\mathcal E$;
\item every vertex belongs to $\mathcal B_{\mathcal E}$ for some frame
$\mathcal E$; and
\item every component is a Boolean subalgebra of
$\mathcal B_{\mathcal E}$ for some frame $\mathcal E$.
\end{enumerate}
\end{definition}

In a finite Boolean algebra, its atoms are pairwise orthogonal and their join
is $1$.  Thus a frame in the preceding definition is exactly the set of atoms
of its associated maximal component.

\subsection{Examples}

\begin{enumerate}
\item
The partial Boolean algebra $(\mathcal V_{12},\Pi^\circ)$ is
three-dimensional.  Its global atoms are
\[
3,\ 5,\ 7,\ 9,\ 11,
\]
and it has exactly two frames,
\[
\mathcal E_1=\{5,7,3\},
\qquad
\mathcal E_2=\{11,7,9\}.
\]
Their associated maximal components are $\mathcal B_1$ and
$\mathcal B_2$, respectively.  Equation~\eqref{eq:p12} shows that every
component is a Boolean subalgebra of one of these two maximal components.

\item
Let $H=\mathbb R^3$, and let $\mathcal V(H)$ be the set of all linear
subspaces of $H$.  Orthocomplementation is the usual orthogonal complement.
Two subspaces are compatible when their orthogonal projections commute;
equivalently, they belong to a common Boolean algebra arising from an
orthogonal decomposition of $H$.

For the standard orthonormal basis, let
$\overline x,\overline y,\overline z$ denote the coordinate axes and let
$\overline{xy},\overline{xz},\overline{yz}$ denote the coordinate planes.
Then
\[
\mathcal B_{\mathcal E}
=
\{0,\overline x,\overline y,\overline z,
  \overline{xy},\overline{xz},\overline{yz},H\}
\]
is a maximal Boolean component with frame
$\mathcal E=\{\overline x,\overline y,\overline z\}$.  For example,
\[
\overline x\sqcup\overline y=\overline{xy},\qquad
\overline{xy}\sqcap\overline{xz}=\overline x,
\qquad
\overline x^\perp=\overline{yz}.
\]
Taking as components all such maximal Boolean algebras and all their Boolean
subalgebras gives a three-dimensional partial Boolean algebra.

\item
More generally, the subspaces of $\mathbb C^n$, with components determined
by orthogonal decompositions (equivalently, by commuting families of
orthogonal projections), form an $n$-dimensional partial Boolean algebra.
This is the finite-dimensional propositional structure used in quantum
theory.
\end{enumerate}

\section{Certainty models associated with pure and mixed quantum states}
\label{pm}

The models introduced in Section~\ref{sec:mal2} are two-valued semantic objects:
they record which initial formul\ae\ are assigned truth value~$1$ and which are
assigned truth value~$0$.
When the core is the partial Boolean algebra of subspaces of a finite-dimensional
Hilbert space, these truth values can be compared with quantum probabilities.
The comparison is exact at probability~$1$: a model records the propositions that
are certain in a quantum state, while grouping all probabilities strictly below~$1$
under truth value~$0$.
Accordingly, the terminology in this section concerns \emph{certainty models
associated with states}, not a reconstruction of the full probability distribution
of a quantum state.
See Appendix~\ref{App:QT} for the precise correspondence.

Let $(\mathcal V,\Pi)$ be an atomic $n$-dimensional partial Boolean algebra, and
let $g:U\twoheadrightarrow\mathcal V$ be a canonical function.
An initial formula $e\in U$ is called \emph{primitive} when $g(e)$ is an atom of
$(\mathcal V,\Pi)$.
For such an $e$, the singleton $\{e\}$ is a cluster and
\begin{equation}
\mathfrak M_e=\{e\}^*
=\{\,q\in U:g(e)\le g(q)\,\}
\label{pure}
\end{equation}
is a model of $\mathcal T_g$.
We call $\mathfrak M_e$ the \emph{pure-state certainty model} based on $e$.
When $\mathcal V$ is a subspace partial Boolean algebra and $g(e)$ is the ray of a
pure state, $\mathfrak M_e$ consists exactly of the propositions having
probability~$1$ in that state.

Now let $a\in U$ satisfy
\[
\mathbf 0<g(a),
\qquad
 g(a)\text{ is not an atom}.
\]
Then
\begin{equation}
\mathfrak M_a=\{a\}^*
=\{\,q\in U:g(a)\le g(q)\,\}
\label{mix}
\end{equation}
is also a model of $\mathcal T_g$.
We call it the \emph{support-state certainty model} based on $a$.
In the Hilbert-space case, if $g(a)$ is the support of a density operator $\rho$,
then $\mathfrak M_a$ is precisely the set of propositions to which $\rho$ assigns
probability~$1$.
Different density operators with the same support determine the same certainty
model, so the model captures the support class of the state rather than its full
probabilistic content.

\subsection{Other models}

Not every model is generated by a singleton cluster.
We describe two further classes.

\begin{enumerate}
\item
\textbf{Hybrid models.}
Suppose $\mathfrak M_1$ and $\mathfrak M_2$ are models such that
$C=\mathfrak M_1\cup\mathfrak M_2$ is a cluster.
Then $C^*$ is a model by Theorem~\ref{thmC}; we call it a \emph{hybrid model}.

For example, with reference to Figure~\ref{fig:v12-pba}, let
\[
\mathfrak M_1=\{p_5\}^*=\{p_5,p_2,p_6,p_1\},
\qquad
\mathfrak M_2=\{p_{10}\}^*=\{p_{10},p_1\}.
\]
The first is a pure-state certainty model, whereas the second is a
support-state certainty model.
Since $g(p_5)\not\perp g(p_{10})$, the set $\{p_5,p_{10}\}$ is a cluster and
\[
\{p_5,p_{10}\}^*
=\{p_1,p_2,p_5,p_6,p_{10}\}
\]
is a hybrid model.
It is not generated by any singleton cluster.

\item
\textbf{Maximal models.}
A model $\mathfrak M$ is \emph{maximal} if it is not a proper subset of another
model.
For example,
\[
\mathfrak M_{5,11}
=\{p_5,p_{11}\}^*
=\{p_1,p_2,p_5,p_6,p_8,p_{11}\}
\]
is maximal.
It is also the union of the two singleton-generated models $\{p_5\}^*$ and
$\{p_{11}\}^*$.

\begin{corollary}
\label{maxmod}
A model $\mathfrak M$ is maximal if and only if, for every
$p\notin\mathfrak M$, there exists $q\in\mathfrak M$ such that $p\perp q$.
\end{corollary}

\begin{proof}
Suppose first that $\mathfrak M$ is maximal and let $p\notin\mathfrak M$.
If $p$ were orthogonal to no element of $\mathfrak M$, then
$\mathfrak M\cup\{p\}$ would be a cluster.
Its extension would be a model containing both $\mathfrak M$ and $p$, contrary
to maximality.

Conversely, suppose every $p\notin\mathfrak M$ is orthogonal to some
$q\in\mathfrak M$.
If $\mathfrak M\subsetneq\mathfrak N$ for a model $\mathfrak N$, choose
$p\in\mathfrak N\setminus\mathfrak M$.
Then $p\perp q$ for some $q\in\mathfrak M\subseteq\mathfrak N$, contradicting
Lemma~\ref{lem:notmodel}.
Hence $\mathfrak M$ is maximal.
\end{proof}
\end{enumerate}

\paragraph{Example.}
The complement of $\mathfrak M_{5,11}$ in $U_{12}$ is
\[
U_{12}\setminus\mathfrak M_{5,11}
=\{p_0,p_3,p_4,p_7,p_9,p_{10}\}.
\]
Each of these formul\ae\ is orthogonal to at least one member of
$\mathfrak M_{5,11}$, as may be checked in Figure~\ref{fig:v12-pba}.
Corollary~\ref{maxmod} therefore gives another proof that
$\mathfrak M_{5,11}$ is maximal.

\section{The Kochen--Specker property}
\label{ksp}

An $n$-dimensional partial Boolean algebra is said to be \emph{nontrivial} if it
has at least two frames.
Let $(\mathcal V,\Pi)$ be a nontrivial $n$-dimensional $\pba$ and let
$\mathfrak A$ denote its set of atoms.

\begin{definition}
\label{def:ks}
The partial Boolean algebra $(\mathcal V,\Pi)$ has the \emph{Kochen--Specker
property} (the \emph{KS-property}) if there exists a function
$f:\mathfrak A\to\{0,1\}$ that assigns the value $1$ to one and only one atom of
\emph{every} frame.
Such a function $f$ is called a \emph{KS-function}.
\end{definition}

\begin{remark}
The terminology in Definition~\ref{def:ks} follows the convention adopted in
this paper: having the KS-property means admitting such a colouring.  In much of
the modern contextuality literature, the phrase ``Kochen--Specker property'' is
instead associated with the obstruction to such a colouring.  Thus our
KS-property is what is often called \emph{KS-colourability}; failure of the
property is the corresponding Kochen--Specker obstruction.
\end{remark}

As always, let $g:U\twoheadrightarrow\mathcal V$ be the canonical function, where
$U$ is the set of initial formul\ae.
Recall (Section~\ref{pm}) that an initial formula $p$ is \emph{primitive} if
$g(p)$ is an atom of $\mathcal V$.

\begin{definition}
\label{def:preframe}
A set $\{p_1,\dots,p_n\}$ of primitive formul\ae\ is a \emph{pre-frame} if
$\{g(p_1),\dots,g(p_n)\}$ is a frame of $(\mathcal V,\Pi)$.
\end{definition}

Since the $n$ atoms of a frame are distinct, $g$ is injective on a pre-frame.
The next lemma is the model-theoretic heart of the matter.

\begin{lemma}
\label{lem:preKS}
If $\mathfrak M$ is a model of $\mathcal T_g$, then $\mathfrak M$ is a cluster and
contains at most one member of any pre-frame.
Consequently, $\mathfrak M\cap\mathcal F$ is either empty or a singleton, for
every pre-frame $\mathcal F$.
\end{lemma}

\begin{proof}
By Theorem~\ref{thm:models-are-clusters}, $\mathfrak M$ is a cluster.
Suppose $p,q\in\mathfrak M\cap\mathcal F$ with $p\ne q$, where $\mathcal F$ is a
pre-frame.
Then $g(p)$ and $g(q)$ are distinct atoms of a common frame, hence mutually
orthogonal, so $p\perp q$.
This contradicts the fact that $\mathfrak M$ is a cluster.
\end{proof}

A schematic picture is helpful: think of $U$ as a disc, the pre-frames as chords,
and the models as solid sub-discs; a model meets each chord in at most one point
(Figure~\ref{fig:models}).

\begin{figure}[ht]
\centering
\includegraphics[width=0.55\textwidth]{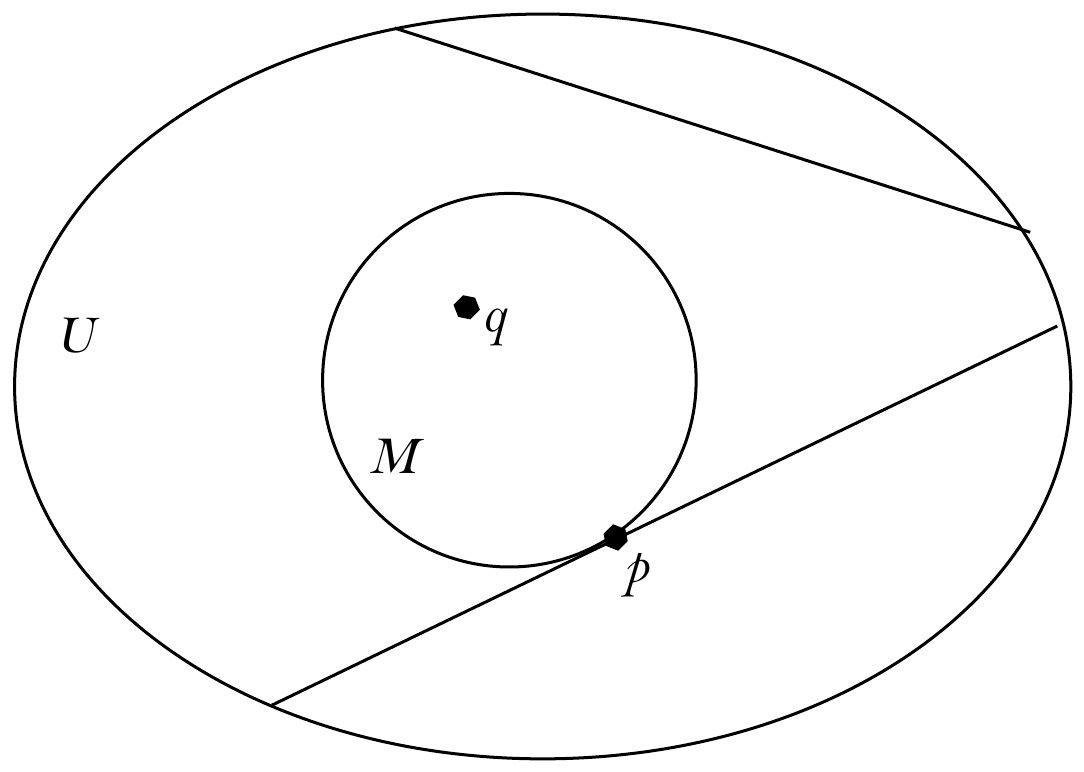}
\caption{\small A schematic representation of models and pre-frames.
The set $U$ of initial formul\ae\ is depicted as a disc, the chords represent
pre-frames, and the solid discs represent models.
A model meets each pre-frame in at most one point, i.e.\ in at most one initial
formula (Lemma~\ref{lem:preKS}).}
\label{fig:models}
\end{figure}

\begin{lemma}
\label{lem:aF}
Let $f$ be a KS-function.
For every pre-frame $\mathcal F$ let $a_{\mathcal F}$ be the unique primitive
formula in $\mathcal F$ with $f\bigl(g(a_{\mathcal F})\bigr)=1$, and set
\[
C=\{\,a_{\mathcal F} : \mathcal F\ \text{is a pre-frame}\,\}.
\]
Then $C$ is a cluster.
\end{lemma}

\begin{proof}
Suppose, to the contrary, that $C$ is not a cluster.
Then there exist distinct $p,q\in C$ with $p\perp q$, so that $g(p)$ and $g(q)$
are mutually orthogonal atoms, while $f(g(p))=f(g(q))=1$.
By Definition~\ref{def:finite-dimensional}(1), every set of mutually orthogonal atoms is contained in
a frame; let $\mathcal E$ be a frame containing both $g(p)$ and $g(q)$.
Then $\mathcal E$ contains two atoms assigned the value $1$ by $f$, contradicting
the assumption that $f$ is a KS-function.
\end{proof}

The KS-property now admits a purely model-theoretic reformulation.

\begin{theorem}
\label{thm:ksmod}
A nontrivial $n$-dimensional partial Boolean algebra $(\mathcal V,\Pi)$ has the
KS-property if and only if there exists a model $\mathfrak M$ of $\mathcal T_g$
such that $\mathfrak M\cap\mathcal F$ is a singleton for every pre-frame
$\mathcal F$.
\end{theorem}

\begin{proof}
($\Leftarrow$) Suppose such a model $\mathfrak M$ exists.
Define $f:\mathfrak A\to\{0,1\}$ by setting $f(a)=1$ if and only if there is a
primitive formula $p\in\mathfrak M$ with $g(p)=a$.
Let $\mathcal E=\{a_1,\dots,a_n\}$ be any frame.
Choose primitive formul\ae\ $p_i$ with $g(p_i)=a_i$ (possible since $g$ is onto
and each $a_i$ is an atom); then $\mathcal F=\{p_1,\dots,p_n\}$ is a pre-frame.
By hypothesis $\mathfrak M\cap\mathcal F=\{p_j\}$ for a unique $j$, so
$f(a_j)=1$.
If $f(a_i)=1$ for some $i\ne j$, there would be a primitive $p_i'\in\mathfrak M$
with $g(p_i')=a_i$; then $p_i'$ and $p_j$ lie in $\mathfrak M$ with
$g(p_i')=a_i\perp a_j=g(p_j)$, contradicting that $\mathfrak M$ is a cluster.
Hence $f$ assigns $1$ to exactly one atom of $\mathcal E$, and $f$ is a
KS-function.

($\Rightarrow$) Suppose $(\mathcal V,\Pi)$ has the KS-property, witnessed by a
KS-function $f$.
Let $C$ and $a_{\mathcal F}$ be as in Lemma~\ref{lem:aF}, so that $C$ is a
cluster, and put $\mathfrak M=C^*$, a model of $\mathcal T_g$ by
Theorem~\ref{thmC}.
For any pre-frame $\mathcal F$ we have $a_{\mathcal F}\in C\subseteq\mathfrak M$,
so $\mathfrak M\cap\mathcal F\ne\emptyset$; by Lemma~\ref{lem:preKS} it is a
singleton.
\end{proof}

\begin{corollary}
\label{cor:noKS}
A nontrivial $n$-dimensional partial Boolean algebra fails to have the
KS-property if and only if, for every model $\mathfrak M$ of $\mathcal T_g$,
there exists a pre-frame $\mathcal F$ with $\mathfrak M\cap\mathcal F=\emptyset$.
\end{corollary}

\subsection{The Kochen--Specker theorem}
\label{subsec:ksthm}

Kochen and Specker \cite[p.~66, Theorem~0]{KS} characterize when a
partial Boolean algebra can be embedded into a Boolean algebra.  In the present
setting, such an embedding may be represented by a family of two-valued
homomorphisms that separates points of $\mathcal V$.  Every two-valued
homomorphism restricts on the atoms of each frame to a KS-function.  The converse
need not follow from the existence of a single KS-function: embeddability
requires a sufficiently rich, point-separating family of such homomorphisms.
When the target Boolean algebra is realized as a field of subsets of a set $S$,
the elements of $S$ may be interpreted as hidden states of a classical phase
space.

For the full partial Boolean algebra of projections of a Hilbert space of
dimension at least three, the Kochen--Specker theorem rules out even one global
two-valued homomorphism of the required kind.  In the terminology of
Definition~\ref{def:ks}, this structure therefore fails to have the KS-property
and, a fortiori, admits no Boolean embedding \cite{KS,Cabello2014}.

The following result shows that the simple example $\mathcal V_{12}$ is, by
contrast, well behaved.

\begin{theorem}
\label{thm:KS12}
The partial Boolean algebra $(\mathcal V_{12},\Pi^\circ)$ has the KS-property.
\end{theorem}

\begin{proof}
By Theorem~\ref{thm:ksmod} it suffices to exhibit a model meeting every pre-frame
in a singleton.
Recall from Section~\ref{sec:finite-pba} that $\mathcal V_{12}$ has exactly two frames,
$\mathcal E_1=\{5,7,3\}$ and $\mathcal E_2=\{11,7,9\}$, and that $g(p_n)=n$ is a
bijection on $U_{12}$, so the pre-frames are
$\mathcal F_1=\{p_5,p_7,p_3\}$ and $\mathcal F_2=\{p_{11},p_7,p_9\}$.
Since $g(p_5)=5$ is not orthogonal to $g(p_{11})=11$ (indeed $5\not\perp 11$,
because $11\not\le 4=5^\perp$), the set $\{p_5,p_{11}\}$ is a cluster, and its
associated model is the maximal model
\[
\mathfrak M_{5,11}=\{p_5\}^*\cup\{p_{11}\}^*
=\{p_1,p_2,p_5,p_6,p_8,p_{11}\}
\]
of Section~\ref{pm}.
Then
\[
\mathfrak M_{5,11}\cap\mathcal F_1=\{p_5\},
\qquad
\mathfrak M_{5,11}\cap\mathcal F_2=\{p_{11}\},
\]
both singletons.
As $\mathcal F_1$ and $\mathcal F_2$ are the only pre-frames, the hypothesis of
Theorem~\ref{thm:ksmod} is satisfied, and $(\mathcal V_{12},\Pi^\circ)$ has the
KS-property.
\end{proof}

\section{Observables and intrinsic formul\ae}
\label{sec:obs}

To connect the abstract framework with the language of measurement, we use
finite-spectrum observables.  This is sufficient for all examples in the paper
and avoids imposing any countable-completeness assumption on an arbitrary
partial Boolean algebra.  Throughout this section, $(\mathcal V,\Pi)$ is an
$n$-dimensional partial Boolean algebra and $\mathbb B$ is the Boolean
$\sigma$-algebra of Borel subsets of $\mathbb R$.

\begin{definition}
\label{def:obs}
An \emph{observable} on $(\mathcal V,\Pi)$ is a map
$P:\mathbb B\to\mathcal V$ for which there exist a frame
$\mathcal F=\{a_1,\dots,a_n\}$ and pairwise distinct real numbers
$\lambda_1,\dots,\lambda_n$ such that
\begin{equation}
\label{eq:finite-observable}
 P(E)=\bigsqcup_{\lambda_i\in E} a_i
 \qquad (E\in\mathbb B),
\end{equation}
where the empty join is $\mathbf 0$.
The finite set
$\spec(P)=\{\lambda_1,\dots,\lambda_n\}$ is the
\emph{spectrum} of $P$, and the Boolean algebra generated by
$\mathcal F$ is the \emph{range component} of $P$.
\end{definition}

Because all $a_i$ lie in one Boolean component, the join in
\eqref{eq:finite-observable} is well defined.  The definition immediately gives
\[
P(\emptyset)=\mathbf 0,
\qquad P(\mathbb R)=\mathbf 1,
\qquad P(E')=P(E)^\perp,
\]
and, for pairwise disjoint Borel sets $E_1,E_2,\dots$,
\[
P\!\left(\bigcup_jE_j\right)=\bigsqcup_jP(E_j).
\]
Only finitely many terms on the right can be nonzero.  Thus a finite-spectrum
observable is a projection-valued measure in the usual finite-dimensional
sense \cite{Varadarajan}.

For a Borel function $\varphi:\mathbb R\to\mathbb R$, define
$\varphi(P)$ by
\[
\varphi(P)(E)=P(\varphi^{-1}(E)).
\]
Its range lies in the range component of $P$.
Two observables $X$ and $Y$ are said to be \emph{compatible}, written
$X\rightleftharpoons Y$, if there are Borel functions $\varphi,\psi$ and an
observable $Z$ such that $X=\varphi(Z)$ and $Y=\psi(Z)$.

\begin{corollary}
\label{cor:comobs}
If $X\rightleftharpoons Y$, then $X(E)\simeq Y(F)$ for all
$E,F\in\mathbb B$.
\end{corollary}

\begin{proof}
The ranges of $X$ and $Y$ lie in the range component of $Z$.  Hence every
$X(E)$ and $Y(F)$ belong to a common Boolean component and are compatible.
\end{proof}

\begin{proposition}
\label{prop:vertex-as-event}
Every vertex $x\in\mathcal V$ is the value $P(E)$ of some observable $P$ at
some Borel set $E$.
\end{proposition}

\begin{proof}
By Definition~\ref{def:finite-dimensional}, there is a frame
$\mathcal F=\{a_1,\dots,a_n\}$ whose maximal component
$\mathcal B_{\mathcal F}$ contains $x$.  Since
$\mathcal B_{\mathcal F}$ is the Boolean algebra generated by its atoms,
there is a subset $I\subseteq\{1,\dots,n\}$ such that
$x=\bigsqcup_{i\in I}a_i$.  Choose pairwise distinct real numbers
$\lambda_1,\dots,\lambda_n$, define $P$ by
\eqref{eq:finite-observable}, and take
$E=\{\lambda_i:i\in I\}$.  Then $P(E)=x$.
\end{proof}

\subsection{Two observables on $\mathcal V_{12}$}
\label{subsec:obsv12}

Consider $(\mathcal V_{12},\Pi^\circ)$ from
Figure~\ref{fig:v12-pba}.  Let
$\mathsf R=\{r_1,r_2,r_3\}$ be a set of three distinct real numbers.
Using the frame $\{5,7,3\}$ of $\mathcal B_1$, define $P_0$ by
\begin{equation}
\label{eq:obs1}
P_0(E)=
\bigsqcup_{r_i\in E} a_i,
\qquad (a_1,a_2,a_3)=(5,7,3).
\end{equation}
Equivalently, its values on subsets of its spectrum are
\[
\begin{array}{c|cccccccc}
S
&\emptyset&\{r_1\}&\{r_2\}&\{r_3\}
&\{r_1,r_2\}&\{r_1,r_3\}&\{r_2,r_3\}&\mathsf R\\ \hline
P_0(S)&0&5&7&3&2&6&4&1.
\end{array}
\]
For an arbitrary Borel set $E$, the value $P_0(E)$ depends only on
$E\cap\mathsf R$.

Likewise, using the frame $\{11,7,9\}$ of $\mathcal B_2$ and distinct real
numbers $\mathsf T=\{t_1,t_2,t_3\}$, define
\begin{equation}
\label{eq:obs2}
Q_0(E)=
\bigsqcup_{t_i\in E} b_i,
\qquad (b_1,b_2,b_3)=(11,7,9).
\end{equation}
Thus
\[
\begin{array}{c|cccccccc}
S
&\emptyset&\{t_1\}&\{t_2\}&\{t_3\}
&\{t_1,t_2\}&\{t_1,t_3\}&\{t_2,t_3\}&\mathsf T\\ \hline
Q_0(S)&0&11&7&9&8&6&10&1.
\end{array}
\]
The observables $P_0$ and $Q_0$ are not compatible.  Indeed,
$P_0(\{r_1,r_2\})=2$ and $Q_0(\{t_3\})=9$, while
$2\not\simeq 9$.  The contrapositive of Corollary~\ref{cor:comobs} therefore
implies $P_0\not\rightleftharpoons Q_0$.

\subsection{Intrinsic and extrinsic formul\ae}
\label{subsec:intrinsic}

An ordered pair $(P,E)$, where $P$ is an observable and $E\in\mathbb B$, is
called an \emph{intrinsic formula}.  Its intended reading is
\begin{center}
\emph{the value of the observable $P$ lies in $E$}.
\end{center}
Let $\mathcal I$ denote the set of all intrinsic formul\ae.  It carries the
canonical map
\begin{equation}
\label{eq:trans}
\widehat g:\mathcal I\longrightarrow\mathcal V,
\qquad
\widehat g(P,E)=P(E).
\end{equation}
Proposition~\ref{prop:vertex-as-event} shows that $\widehat g$ is surjective.
Consequently it induces a consistent theory $\mathcal T_{\widehat g}$ by the
construction of Section~\ref{sec:mal2}.  Among its axioms are
\begin{enumerate}
\item[($\alpha_1$)] $(P,\mathbb R)$;
\item[($\alpha_2$)] $\neg(P,\emptyset)$;
\item[($\alpha_3$)] $(P,E)\to(Q,F)$ whenever $P(E)\le Q(F)$;
\item[($\alpha_4$)] $(P,E)\to\neg(Q,F)$ whenever $P(E)\perp Q(F)$.
\end{enumerate}
In particular, $E\subseteq F$ implies
$(P,E)\to(P,F)$, and $E\cap F=\emptyset$ implies
$(P,E)\to\neg(P,F)$.

An initial formula not presented as a pair $(P,E)$ is called an
\emph{extrinsic formula}.  The next definition relates the two presentations.

\begin{definition}
\label{def:syn}
Let $U$ be a set of extrinsic initial formul\ae\ with canonical surjection
$g:U\twoheadrightarrow\mathcal V$.  An intrinsic formula $(P,E)$ is
\emph{synonymous} with $p\in U$, written
$p\sim_{\mathrm{syn}}(P,E)$, if
\[
g(p)=P(E).
\]
\end{definition}

By Proposition~\ref{prop:vertex-as-event}, every extrinsic initial formula has
at least one intrinsic synonym.  To express logical equivalence, work in the
combined set of initial formul\ae
$\mathcal H=U\sqcup\mathcal I$ and define
$g^*:\mathcal H\to\mathcal V$ by
$g^*|_U=g$ and $g^*|_{\mathcal I}=\widehat g$.
Theorem~\ref{thm:equiv} then yields
\[
g^*(x)=g^*(y)
\quad\Longleftrightarrow\quad
(x\leftrightarrow y)\in\mathcal T_{g^*}.
\]
Hence synonymous initial formul\ae\ are logically equivalent in the combined
theory.  Synonymy does not identify the syntactic expressions themselves; it
identifies their common image in the core.

\subsection{A lexicon for $\mathcal V_{12}$}
\label{sec:lexicon}

For $U_{12}$, $g$, $P_0$, and $Q_0$, the following are convenient
spectrum-reduced representatives of the synonymy classes:
\begin{align*}
&p_0\sim_{\mathrm{syn}}(P_0,\emptyset)
      \sim_{\mathrm{syn}}(Q_0,\emptyset),
&&p_6\sim_{\mathrm{syn}}(P_0,\{r_1,r_3\})
      \sim_{\mathrm{syn}}(Q_0,\{t_1,t_3\}),\\
&p_1\sim_{\mathrm{syn}}(P_0,\mathsf R)
      \sim_{\mathrm{syn}}(Q_0,\mathsf T),
&&p_7\sim_{\mathrm{syn}}(P_0,\{r_2\})
      \sim_{\mathrm{syn}}(Q_0,\{t_2\}),\\
&p_2\sim_{\mathrm{syn}}(P_0,\{r_1,r_2\}),
&&p_8\sim_{\mathrm{syn}}(Q_0,\{t_1,t_2\}),\\
&p_3\sim_{\mathrm{syn}}(P_0,\{r_3\}),
&&p_9\sim_{\mathrm{syn}}(Q_0,\{t_3\}),\\
&p_4\sim_{\mathrm{syn}}(P_0,\{r_2,r_3\}),
&&p_{10}\sim_{\mathrm{syn}}(Q_0,\{t_2,t_3\}),\\
&p_5\sim_{\mathrm{syn}}(P_0,\{r_1\}),
&&p_{11}\sim_{\mathrm{syn}}(Q_0,\{t_1\}).
\end{align*}
These representatives are not unique: replacing an event by any Borel set
having the same intersection with the relevant spectrum gives the same vertex.
The vertices $0,1,6,$ and $7$ also lie in both maximal components and therefore
have representatives involving either $P_0$ or $Q_0$.

\section{A macroscopic world}
\label{macw}

In the spirit of introductory texts on classical logic \cite{Copi1979}, let the
initial formul\ae\ be statements about characters in a fictitious macroscopic
world.  For example, set
\begin{quote}
$p_0=$ ``Rover, the dog, is aggressive''; \quad
$p_1=$ ``Peter loves football''; \quad
$p_2=$ ``Alice, the cat, is boisterous'';\\
$p_3=$ ``Alice, the cat, is hungry''; \quad
$p_4=$ ``Mother prepares a meal''; \quad
$p_5=$ ``Rover looks for somewhere to hide''.
\end{quote}
The remaining $p_n$ may be assigned any further statements suitable to the
story.  With $g(p_n)=n$ as in~\eqref{spkl1}, the complete axiom set is the one
induced by clauses (g1)--(g4) of Section~\ref{sec:mal2}.  A few representative
axioms are:
\begin{enumerate}
\item[(A1)] $\neg p_0$: ``Rover is not aggressive,'' since
$g(p_0)=\mathbf 0$;
\item[(A2)] $p_1$: ``Peter loves football,'' since
$g(p_1)=\mathbf 1$;
\item[(A3)] $p_5\to p_2$: ``If Rover looks for somewhere to hide, then Alice is
boisterous,'' since $g(p_5)\le g(p_2)$;
\item[(A4)] $p_5\to\neg p_3$: ``If Rover looks for somewhere to hide, then Alice
is not hungry,'' since $g(p_5)\perp g(p_3)$;
\item[(A5)] $p_3\to p_4$: ``If Alice is hungry, then Mother prepares a meal,''
since $g(p_3)\le g(p_4)$;
\item[(A6)] $p_8\to\neg p_9$, since $g(p_8)\perp g(p_9)$;
\item[(A7)] $p_9\to p_{10}$, since $g(p_9)\le g(p_{10})$;
\item[(A8)] $p_{11}\to\neg p_{10}$, since
$g(p_{11})\perp g(p_{10})$.
\end{enumerate}
Let $\mathcal T_{12}$ denote the theory generated by the \emph{full} axiom set
induced by $g$, not merely by the displayed sample.  By
Theorem~\ref{thmC}, it is consistent, and by
Theorem~\ref{thm:corepba}, its core is isomorphic to
$(\mathcal V_{12},\Pi^\circ)$.  Thus an entirely classical propositional theory
can have a non-Boolean core even when its initial formul\ae\ describe a
macroscopic fiction.  The example is an algebraic illustration; by itself it
does not assert that the fictional household obeys a physical quantum theory.

If $\sigma$ is a permutation of $\{0,\dots,11\}$ and $g$ is replaced by
$\check g(p_n)=\sigma(n)$, one obtains another consistent theory whose core is
isomorphic to the same partial Boolean algebra, but whose English rules may be
more or less plausible.

\section{Heuristics of truth and quantum-like uncertainty}
\label{fuzzy}

We have seen in Theorem~\ref{notneg} that orthocomplementation is not
negation: if $g(p)=g(q)^\perp$, then $p$ and $q$ cannot be true together,
but they can both be false.  In the language of intrinsic formul\ae\ this
becomes
\begin{center}
$(P,E)$ and $(P,E')$ cannot both be true, but they can both be false.
\end{center}
At first sight this conflicts with the familiar picture of an ideal
measurement pointer selecting a single value of $P$, which must lie either in
$E$ or in its complement $E'$.

The resolution suggested in \cite{Malhas1993} is to allow an
\emph{extended}, or \emph{fuzzy}, pointer.  Such a pointer need not select one
real number; it may determine only a set of values compatible with the current
model.  This leads to a model-theoretic notion of uncertainty that does not
presuppose a probability measure and is not attributed to imperfect
instrumentation.

\begin{definition}
\label{def:pointer}
Let $\mathfrak M$ be a model of $\mathcal T_{\widehat g}$ and let $P$ be an
observable with finite spectrum $\spec(P)$.  The \emph{pointer set} of
$P$ in $\mathfrak M$ is
\[
\Delta_{\mathfrak M}(P)
 =\bigcap\{\,E\cap\spec(P):(P,E)\in\mathfrak M\,\}.
\]
We say that $P$ is \emph{actualizable} in $\mathfrak M$ when
$\Delta_{\mathfrak M}(P)\ne\emptyset$, and in that case the
\emph{uncertainty} of $P$ in $\mathfrak M$ is the nonnegative integer
\[
\|\Delta_{\mathfrak M}(P)\|
 =|\Delta_{\mathfrak M}(P)|-1 .
\]
We say that $P$ is \emph{sharp} in $\mathfrak M$ when
$|\Delta_{\mathfrak M}(P)|=1$, that is, when $P$ is actualizable with
uncertainty $0$.
\end{definition}

The defining family is nonempty because $(P,\spec(P))$ is true in every
model.  Since the spectrum is finite, the intersection is well defined.  It
is important, however, that $\Delta_{\mathfrak M}(P)$ is a \emph{derived
summary} of the true $P$-statements.  In general, the model axioms give upward
closure under implication but not closure under intersections; therefore one
must not assume that
$(P,\Delta_{\mathfrak M}(P))$ itself belongs to $\mathfrak M$.

\begin{remark}[Why actualizability must be assumed]
\label{rem:empty-pointer}
The restriction to actualizable observables is not vacuous: a model may make
the pointer set of an observable \emph{empty}, in which case $|\Delta|-1=-1$
and the expression would not measure anything.  Take
$(\mathcal V_{12},\Pi^\circ)$ with $g(p_n)=n$ and let
$C=\{p_2,p_4,p_6\}$, which is one of the clusters exhibited in
Section~\ref{sec:clusters}.  The associated model is
\[
C^*=\{p_1,p_2,p_4,p_6\},
\]
and the true $(P_0,\cdot)$ formul\ae\ correspond to the events
$\{r_1,r_2\}$, $\{r_2,r_3\}$, $\{r_1,r_3\}$ and $\mathsf R$, whose
intersection is empty.  Thus $\Delta_{C^*}(P_0)=\emptyset$: the model asserts
three two-valued restrictions on $P_0$ that are pairwise, but not jointly,
satisfiable by a single eigenvalue.  Of the $53$ models of $\mathcal T_{12}$,
exactly $11$ leave one of $P_0,Q_0$ non-actualizable.  Such a model carries
consistent information about $P_0$---it is a genuine model of a consistent
theory---but no single outcome of $P_0$ is compatible with all of it, and the
measurement-update rule of Section~\ref{sec:dynamics} is correspondingly
undefined for $P_0$ there.  This is the model-theoretic trace of the fact that
$\mathfrak M_P$ need not be closed under meets.
\end{remark}

\subsection{A simple example}
\label{subsec:simpleunc}

Take $(\mathcal V_{12},\Pi^\circ)$, $U_{12}$, $g$, $P_0$, and $Q_0$ as above,
and consider the pure-state model based on the primitive formula
$p_5=(P_0,\{r_1\})$:
\[
\mathfrak M_5=\{p_5\}^*=\{p_5,p_2,p_6,p_1\},
\]
which corresponds to the vertices $\{5,2,6,1\}$.
Using the lexicon of Section~\ref{sec:lexicon}, the true intrinsic formul\ae\
of the form $(P_0,\cdot)$ correspond to
\[
\{r_1\},\qquad \{r_1,r_2\},\qquad \{r_1,r_3\},\qquad
\{r_1,r_2,r_3\}.
\]
Their intersection is $\{r_1\}$, and therefore
\[
\Delta_{\mathfrak M_5}(P_0)=\{r_1\},
\qquad
\|\Delta_{\mathfrak M_5}(P_0)\|=0.
\]
Thus $P_0$ is sharp in $\mathfrak M_5$.

The non-commuting observable $Q_0$ behaves differently.  The true
$(Q_0,\cdot)$ formul\ae\ correspond to
\[
\{t_1,t_3\}\quad(\text{vertex }6),
\qquad
\{t_1,t_2,t_3\}\quad(\text{vertex }1).
\]
Hence
\[
\Delta_{\mathfrak M_5}(Q_0)=\{t_1,t_3\},
\qquad
\|\Delta_{\mathfrak M_5}(Q_0)\|=1.
\]
The value of $Q_0$ is therefore not sharply determined in this model.  This is
\emph{model-relative} uncertainty: it cannot be removed while the model
$\mathfrak M_5$ is held fixed, although later sections show that the algebra
$\mathcal V_{12}$ as a whole admits dispersion-free refinements.

\subsection{Uncertainty as a count of locally admissible outcomes}
\label{subsec:outcomes}

Let $P$ have finite spectrum and let $\mathcal B_P$ be its range component
(Definition~\ref{def:obs}), a finite Boolean algebra.  Write $\mathcal F_P$ for
the frame of atoms of $\mathcal B_P$.  If $a_r=P(\{r\})$ is the atom
corresponding to an eigenvalue $r$, define the $P$-part of the model by
\[
\mathfrak M_P
 =\{\,P(E):(P,E)\in\mathfrak M\,\}\subseteq\mathcal B_P.
\]

\begin{definition}
\label{def:outcomes}
The set of \emph{locally admissible outcomes} of $P$ in $\mathfrak M$ is
\[
\Outloc{P}{\mathfrak M}
 =\{\,a\in\mathcal F_P:
      a\text{ is orthogonal to no element of }\mathfrak M_P\,\}.
\]
Equivalently, $a\in\Outloc{P}{\mathfrak M}$ when
$a\sqcap x\neq\mathbf 0$ for every $x\in\mathfrak M_P$; equivalently again,
since $a$ is an atom of the finite Boolean algebra $\mathcal B_P$, when
$a\le x$ for every $x\in\mathfrak M_P$.
\end{definition}

Because $\mathcal B_P$ is finite, the last formulation says exactly that
$\Outloc{P}{\mathfrak M}$ is the set of atoms below
$\bigsqcap\mathfrak M_P$.  In particular
\[
\bigl|\Outloc{P}{\mathfrak M}\bigr|=0
\quad\Longleftrightarrow\quad
\bigsqcap\mathfrak M_P=\mathbf 0,
\]
which is the situation of Remark~\ref{rem:empty-pointer}.

The adjective ``local'' is essential: this definition tests compatibility
with the information about $P$ contained in the Boolean component
$\mathcal B_P$.  It does not assert that orthogonality to the whole model is
determined by $\mathfrak M\cap\mathcal B_P$.

\begin{proposition}
\label{prop:outcount}
For every model $\mathfrak M$ and observable $P$ with finite spectrum, the
eigenvalues in $\Delta_{\mathfrak M}(P)$ are exactly those whose atoms lie in
$\Outloc{P}{\mathfrak M}$.  Consequently $P$ is actualizable in
$\mathfrak M$ if and only if $\Outloc{P}{\mathfrak M}\ne\emptyset$, and in
that case
\[
\|\Delta_{\mathfrak M}(P)\|
 =\bigl|\Outloc{P}{\mathfrak M}\bigr|-1.
\]
\end{proposition}

\begin{proof}
Let $a_r=P(\{r\})$ be an atom of $\mathcal B_P$.  Since $\mathcal B_P$ is a
finite Boolean algebra, for every $x=P(E)\in\mathcal B_P$ one has
\[
a_r\not\perp x
\quad\Longleftrightarrow\quad
 a_r\le x
\quad\Longleftrightarrow\quad
 r\in E.
\]
Therefore $a_r$ is orthogonal to no element of $\mathfrak M_P$ if and only if
$r$ belongs to every $E$ for which $(P,E)\in\mathfrak M$.  This is precisely
the condition $r\in\Delta_{\mathfrak M}(P)$.
\end{proof}

For the model $\mathfrak M_5$ one obtains
\[
\Outloc{{P_0}}{\mathfrak M_5}=\{5\},
\qquad
\Outloc{{Q_0}}{\mathfrak M_5}=\{9,11\},
\]
which reproduces the uncertainties $0$ and $1$ computed above.

\section{Measurement dynamics and back-action}
\label{sec:dynamics}

The preceding notion of uncertainty is static.  To discuss successive
measurements, an additional update rule is required.  We now adopt a natural
rule and then derive its consequences from the cluster structure.  Thus the
update itself is a modelling postulate, whereas repeatability and the
back-action calculations below follow from the theory once that rule is fixed.

The pointer analysis of Section~\ref{fuzzy} is local to the Boolean component
of the measured observable.  An actual update must also remain consistent with
the whole model.  This motivates a second, global notion of admissibility.

\begin{definition}
\label{def:global-outcomes}
Let $P$ have frame $\mathcal F_P$.  The set of \emph{globally admissible
outcomes} of $P$ in a model $\mathfrak M$ is
\[
\Outglob{P}{\mathfrak M}
 =\{\,a\in\mathcal F_P:
      a\text{ is orthogonal to no }g(q),\ q\in\mathfrak M\,\}.
\]
Equivalently, if $e_a$ is any primitive formula with $g(e_a)=a$, then
$a\in\Outglob{P}{\mathfrak M}$ exactly when
$\mathfrak M\cup\{e_a\}$ is a cluster.
\end{definition}

For the intrinsic language one may take $e_a=(P,\{r\})$ when
$a=P(\{r\})$.  Global admissibility implies that the proposed outcome does not
contradict any formula already true in the model.  Clearly
\[
\Outglob{P}{\mathfrak M}\subseteq\Outloc{P}{\mathfrak M},
\]
since $\mathfrak M_P$ is part of the $g$-image of $\mathfrak M$; the two sets
need not coincide in an arbitrary partial Boolean algebra, though in the finite
$\mathcal V_{12}$ calculations below they agree for the outcomes under
consideration.  Both may be empty, by Remark~\ref{rem:empty-pointer}; the
update rule of Definition~\ref{def:update} is then simply not applicable to
$P$ in $\mathfrak M$.

\begin{definition}
\label{def:update}
Suppose a measurement of $P$ in $\mathfrak M$ returns a globally admissible
outcome $a\in\Outglob{P}{\mathfrak M}$.  We postulate that
the updated model is the pure-state model based on $a$:
\[
\mathfrak M\ \xmapsto{\ P,\,a\ }\ \{e_a\}^*
 =\{\,q\in U:a\le g(q)\,\}.
\]
This map is called the \emph{measurement-update map}.
\end{definition}

The update is well defined because $a$ is an atom, hence $\{e_a\}$ is a
cluster, and Theorem~\ref{thmC} implies that $\{e_a\}^*$ is a model.

\begin{proposition}[Repeatability]
\label{prop:repeat}
Let $a$ be an outcome atom of $P$.  In the updated model $\{e_a\}^*$, the
observable $P$ is sharp with pointer set equal to the eigenvalue corresponding
to $a$.  Moreover,
\[
\Outglob{P}{\{e_a\}^*}=\{a\}.
\]
Consequently, an immediate repetition of the same measurement returns the same
outcome and leaves the model unchanged.
\end{proposition}

\begin{proof}
The formula $e_a$ belongs to $\{e_a\}^*$, so the singleton event corresponding
to $a$ is true and the pointer of $P$ is sharp.  Every other atom $b$ of the
same frame is orthogonal to $a$, and hence is not globally admissible in a
model containing $e_a$.  The atom $a$ is not orthogonal to any vertex above it,
so it remains admissible.  Applying the update rule again therefore returns
$\{e_a\}^*$.
\end{proof}

The effect on an incompatible observable can be read explicitly from
$\mathcal V_{12}$.

\begin{theorem}[Back-action in $\mathcal V_{12}$]
\label{thm:backaction}
Let $P_0$ and $Q_0$ be the observables of
Section~\ref{subsec:obsv12}, with frames
\[
\mathcal F_{P_0}=\{5,7,3\},
\qquad
\mathcal F_{Q_0}=\{11,7,9\}.
\]
After a $Q_0$-measurement with outcome $a$, the local pointer of $P_0$ in the
updated model is given by
\[
\Outloc{{P_0}}{\{a\}^*}=
\begin{cases}
\{7\}, & a=7,\\
\{3,5\}, & a=9,\\
\{3,5\}, & a=11.
\end{cases}
\]
Thus $P_0$ remains sharp exactly for the shared atom $7$; either non-shared
outcome makes $P_0$ uncertain with uncertainty equal to $1$.
\end{theorem}

\begin{proof}
For $a=7$,
\[
\{7\}^*=\{7,2,4,8,10,1\}.
\]
Its $\mathcal B_1$-part contains the atom $7$, so the only locally admissible
$P_0$-outcome is $7$.

For $a=9$,
\[
\{9\}^*=\{9,10,6,1\}.
\]
The intersection with $\mathcal B_1$ is $\{6,1\}$.  In $\mathcal B_1$ the
coatom $6$ contains precisely the atoms $3$ and $5$, and is orthogonal to the
atom $7$.  Hence
$\Outloc{{P_0}}{\{9\}^*}=\{3,5\}$.

For $a=11$,
\[
\{11\}^*=\{11,8,6,1\},
\]
and again the $\mathcal B_1$-part is $\{6,1\}$.  The same calculation gives
$\Outloc{{P_0}}{\{11\}^*}=\{3,5\}$.
\end{proof}

\paragraph{Worked transition.}
Start from
$\mathfrak M_5=\{5,2,6,1\}$, where $P_0$ is sharp with value $r_1$ and
\[
\Outglob{{Q_0}}{\mathfrak M_5}=\{9,11\}.
\]
Suppose the $Q_0$-measurement returns $9=(Q_0,\{t_3\})$.  The update rule gives
\[
\mathfrak M_5\ \xmapsto{\ Q_0,9\ }\ \{9\}^*=\{9,10,6,1\}.
\]
The observable $Q_0$ is now sharp, whereas
\[
\Outloc{{P_0}}{\{9\}^*}=\{3,5\},
\qquad
\|\Delta_{\{9\}^*}(P_0)\|=1.
\]
Thus the previously sharp $P_0$ becomes uncertain under the adopted update rule.
The update rule is postulated, but the resulting loss of sharpness is forced by
the order and compatibility relations of the partial Boolean algebra.

\begin{remark}
\label{rem:nondet}
The update rule does not assign probabilities to its possible outcomes.  When
more than one globally admissible outcome exists, the abstract theory specifies
the possible successor models but not their relative frequencies.  In a
Hilbert-space realisation, probabilities may be supplied by the usual quantum
state formalism; in the purely combinatorial setting, constructing a probability
rule is a separate problem.
\end{remark}

\section{Hidden variables and the elimination of uncertainty}
\label{sec:hv}

Whether the uncertainties of Section~\ref{fuzzy} can be ``explained away'' by
underlying hidden states depends precisely on whether the core has the
KS-property.
We illustrate the favourable case using $\mathcal V_{12}$.

Recall that $(\mathcal V_{12},\Pi^\circ)$ is $3$-dimensional with exactly two
frames, $\mathcal E_1=\{5,7,3\}$ and $\mathcal E_2=\{11,7,9\}$.
By Theorem~\ref{thm:KS12} it has the KS-property; in fact there are exactly five
KS-functions.  Writing each one by its set of atoms of value~$1$, they are
\[
\{7\},\quad\{5,11\},\quad\{5,9\},\quad\{3,11\},\quad\{3,9\}:
\]
either the shared atom $7$ is selected in both frames, or a non-shared atom is
chosen independently in each frame, giving $1+2\cdot2=5$ possibilities.
Consider, for instance, the assignment
\[
s(5)=1,\quad s(7)=0,\quad s(3)=0,\quad s(11)=1,\quad s(9)=0,
\]
which assigns $1$ to exactly one atom of each frame.
This KS-function extends to a two-valued function
$\overline s:\mathcal V_{12}\to\{0,1\}$ by
\[
\overline s(1)=\overline s(5)=\overline s(11)=\overline s(2)
=\overline s(8)=\overline s(6)=1,
\]
\[
\overline s(0)=\overline s(7)=\overline s(3)=\overline s(9)
=\overline s(4)=\overline s(10)=0 .
\]
A direct inspection of~\eqref{eq:p12}, using Figure~\ref{fig:v12-pba}, shows that the
restriction of $\overline s$ to every component of $(\mathcal V_{12},\Pi^\circ)$
is a Boolean homomorphism; hence $\overline s$ is a two-valued homomorphism on the
partial Boolean algebra, a simple instance of the homomorphisms of \cite{KS}.

Writing down all five such homomorphisms, one verifies that for any two distinct
vertices $x,y$ of $\mathcal V_{12}$ there is one among them that \emph{separates}
$x$ and $y$, sending one to $0$ and the other to $1$.
By the Kochen--Specker embedding theorem (Section~\ref{subsec:ksthm}), it follows
that $\mathcal V_{12}$ embeds into the Boolean algebra of all subsets of a set
$S$, which plays the role of the phase space of a classical system whose hidden
states are the elements of $S$.
In this macroscopic world, then, the quantum-like uncertainties of
Section~\ref{fuzzy} are only \emph{apparent}: they reflect ignorance of the
hidden state, and can be eliminated in principle by determining that state.

\subsection{Dispersion-free models}
\label{subsec:dispfree}

The phrase ``eliminated in principle'' should refer to a global assignment, not
merely to the sharpness of one chosen observable.  The pre-frame formulation of
Section~\ref{ksp} provides the appropriate definition.

\begin{definition}
\label{def:dispersion-free}
A model $\mathfrak M$ of $\mathcal T_g$ is \emph{dispersion-free} if
\[
|\mathfrak M\cap\mathcal F|=1
\]
for every pre-frame $\mathcal F$ of primitive formul\ae.
\end{definition}

Thus a dispersion-free model selects exactly one primitive formula over every
frame.  This is stronger and more precise than requiring a singleton pointer for
an arbitrarily chosen family of observables.  It also connects directly with the
model-theoretic characterization of the KS-property proved in
Theorem~\ref{thm:ksmod}.

\begin{proposition}
\label{prop:disp-sharp}
Every dispersion-free model makes every finite-spectrum observable sharp in the
sense of Definition~\ref{def:pointer}.
\end{proposition}

\begin{proof}
Let $P$ have frame $\mathcal F_P$, and let
\[
\mathcal F=\{\,(P,\{r\}):r\in\spec(P)\,\}
\]
be the corresponding pre-frame of primitive intrinsic formul\ae.  By dispersion
freeness, $\mathfrak M\cap\mathcal F$ consists of one formula
$(P,\{r_0\})$.  Hence $r_0$ belongs to every true $P$-event, while no other
eigenvalue can belong to the pointer intersection because the singleton event
$\{r_0\}$ itself is true.  Therefore
$\Delta_{\mathfrak M}(P)=\{r_0\}$.
\end{proof}

The converse of Proposition~\ref{prop:disp-sharp} fails: a model may make every
observable sharp without being dispersion-free, because sharpness of $P$ asserts
only that $\bigsqcap\mathfrak M_P$ is an atom, not that this atom is
\emph{itself} true in $\mathfrak M$.  In $\mathcal V_{12}$, for instance, the
model $\{p_1,p_2,p_6,p_8\}$ makes both $P_0$ and $Q_0$ sharp yet meets neither
pre-frame; of the $53$ models of $\mathcal T_{12}$, exactly $5$ are
dispersion-free while $18$ make both observables sharp.  What survives is the
following weaker but sufficient implication, which is what the applications
actually need.

\begin{theorem}[Universal sharpness forces KS-colourability]
\label{thm:sharp-implies-ks}
Let $(\mathcal V,\Pi)$ be a nontrivial finite-dimensional partial Boolean
algebra and let $\mathfrak M$ be a model of $\mathcal T_g$ in which
\emph{every} finite-spectrum observable is sharp.  Then $(\mathcal V,\Pi)$ has
the KS-property.
\end{theorem}

\begin{proof}
Write $h(v)=1$ when some (equivalently, by Theorem~\ref{thm:equiv}, every)
$q\in U$ with $g(q)=v$ lies in $\mathfrak M$, and identify $\mathfrak M$ with
the up-closed set $\{v\in\mathcal V:h(v)=1\}$, which is a cluster by
Theorem~\ref{thm:models-are-clusters}.

Let $\mathcal E$ be a frame with maximal component $\mathcal B_{\mathcal E}$,
and let $P_{\mathcal E}$ be an observable with that range component
(Proposition~\ref{prop:vertex-as-event}).  By hypothesis and
Proposition~\ref{prop:outcount},
$\Outloc{P_{\mathcal E}}{\mathfrak M}$ is a singleton, say
$\{a_{\mathcal E}\}$ with $a_{\mathcal E}\in\mathcal E$.

\emph{Claim: if $a_{\mathcal E}$ also belongs to a frame $\mathcal E'$, then
$a_{\mathcal E'}=a_{\mathcal E}$.}
Write $a=a_{\mathcal E}$ and suppose $a\notin
\Outloc{P_{\mathcal E'}}{\mathfrak M}$.  Then some
$v\in\mathfrak M\cap\mathcal B_{\mathcal E'}$ satisfies $a\perp v$, that is,
$v\le a^{\perp}$.  Since $\mathfrak M$ is upward closed and $g$ is onto, it
follows that $a^{\perp}\in\mathfrak M$.  But $a\in\mathcal E$ gives
$a^{\perp}\in\mathcal B_{\mathcal E}$, so $a^{\perp}$ lies in
$\mathfrak M_{P_{\mathcal E}}$, while $a\perp a^{\perp}$.  This contradicts
$a\in\Outloc{P_{\mathcal E}}{\mathfrak M}$.  Hence $a$ is locally admissible
for $\mathcal E'$ as well, and since that set is the singleton
$\{a_{\mathcal E'}\}$ we get $a_{\mathcal E'}=a$, proving the claim.

Now define $f:\mathfrak A\to\{0,1\}$ by $f(a)=1$ if and only if
$a=a_{\mathcal E}$ for some frame $\mathcal E$ containing $a$.  By the claim,
$f(a)=1$ if and only if $a=a_{\mathcal E}$ for \emph{every} frame $\mathcal E$
containing $a$.  Consequently each frame $\mathcal E$ contains exactly one atom
of value~$1$, namely $a_{\mathcal E}$, so $f$ is a KS-function.
\end{proof}

\begin{theorem}
\label{thm:dispfree}
Let $(\mathcal V,\Pi)$ be a nontrivial finite-dimensional partial Boolean
algebra and let $\mathcal T_g$ be the induced theory.  The following are
equivalent.
\begin{enumerate}
\item[(i)] $(\mathcal V,\Pi)$ has the KS-property in the terminology of this
paper; that is, its atoms admit a function assigning $1$ to exactly one atom of
every frame.
\item[(ii)] $\mathcal T_g$ has a dispersion-free model.
\item[(iii)] $(\mathcal V,\Pi)$ admits a two-valued homomorphism whose
restriction to every component is a Boolean homomorphism.
\end{enumerate}
\end{theorem}

\begin{proof}
(i)$\Rightarrow$(ii).
This is Theorem~\ref{thm:ksmod}: a KS-function determines a model meeting every
pre-frame in a singleton, which is dispersion-free by
Definition~\ref{def:dispersion-free}.

(ii)$\Rightarrow$(iii).
Let $\mathfrak M$ be dispersion-free.  Because formulas with the same image
under $g$ are equivalent modulo $\mathcal T_g$, membership of a vertex in the
image of $\mathfrak M$ is unambiguous.  Define
\[
h(x)=1
\quad\Longleftrightarrow\quad
\text{some (equivalently every) }p\in U\text{ with }g(p)=x
\text{ belongs to }\mathfrak M.
\]
Let $\mathcal B$ be a component with frame of atoms
$\{a_1,\dots,a_n\}$.  Dispersion freeness selects exactly one corresponding
primitive formula, say the one over $a_j$.  Since models are upward closed,
$h(x)=1$ for every $x\in\mathcal B$ with $a_j\le x$.  If $a_j\not\le x$, then
$x\le a_j^\perp$, so a formula representing $x$ is orthogonal to the selected
primitive formula and cannot lie in the cluster $\mathfrak M$.  Hence
$h\restriction\mathcal B$ is exactly the two-valued Boolean homomorphism
concentrated at $a_j$.  Since this holds for every component, $h$ is a
two-valued homomorphism on the partial Boolean algebra.

(iii)$\Rightarrow$(i).
Restrict the homomorphism to the atoms.  On each Boolean component it assigns
$1$ to exactly one atom of the component's frame and $0$ to all the others.
Thus its restriction is a KS-function.
\end{proof}

When these equivalent conditions hold, all finite-spectrum observables are
simultaneously sharp in the corresponding dispersion-free model by
Proposition~\ref{prop:disp-sharp}.  When they fail, no model meets every
pre-frame in a singleton, and moreover, by
Theorem~\ref{thm:sharp-implies-ks}, \emph{no} model makes every observable
sharp: in every model at least one observable is either non-actualizable or
carries uncertainty at least~$1$.  The obstruction is therefore global rather
than a feature of one selected observable.

For $\mathcal V_{12}$, the two-valued homomorphism $\overline s$ displayed above
selects the shared atom $7$.  The corresponding dispersion-free model is
\[
\{7\}^*=\{7,2,4,8,10,1\},
\]
in which both $P_0$ and $Q_0$ are sharp with the common value $r_2=t_2$.
The five KS-functions give five global value patterns.  This is the colourable
example promised in the abstract; the later $140$-vertex construction provides
the contrasting obstruction.

\section{A finite four-dimensional partial Boolean algebra}
\label{sec:4dim}

We now construct a finite $4$-dimensional partial Boolean algebra that does
\emph{not} have the KS-property.  The construction is combinatorial.  Since the
quotient used below identifies words lying on different Boolean components, the
main point is not merely to count the resulting vertices, but also to verify that
congruence is an equivalence relation and that the Boolean operations induced on
its equivalence classes are independent of all choices of representatives.

We should say at once what is and is not new here.  The underlying incidence
structure---twenty-four points arranged in twenty-four four-element blocks, each
point lying in exactly four blocks---is not new: it is isomorphic to the Peres
configuration of $24$ rays in $\mathbb R^{4}$ and their $24$ orthogonal bases
\cite{Peres1991}, and the nine blocks used in the parity argument of
Section~\ref{subsec:sigma} reproduce the $18$-vector, $9$-basis proof of Cabello,
Estebaranz and Garc\'ia-Alcaine \cite{Cabello1996}, which is one of several
parity proofs contained in the Peres system.  What is new is the passage from a
ray configuration to a \emph{partial Boolean algebra presented purely by words
and side-complementation}, the resulting $140$-vertex quotient, and---by
Theorem~\ref{thm:corepba}---its realization as the core of a consistent classical
propositional theory, so that the geometric obstruction becomes a statement about
which two-valued models that theory admits.  See
Remark~\ref{rem:peres} below for the explicit correspondence.

\subsection{The configuration $\Sigma$}
\label{subsec:sigma}

A \emph{module} consists of eight letters
$S_1,S_2,S_3,S_4,T_1,T_2,T_3,T_4$ and four four-element \emph{sides}
\[
 l=\{S_1,S_2,S_3,S_4\},\qquad
 r=\{T_1,T_2,T_3,T_4\},
\]
\[
 k_1=\{S_1,T_1,T_2,S_2\},\qquad
 k_2=\{S_3,T_3,T_4,S_4\}.
\]
Thus $l$ and $r$ are the left and right columns, while $k_1$ and $k_2$ are
the top and bottom loops in Figure~\ref{module}.

\begin{figure}[ht]
\centering
\includegraphics[scale=0.37]{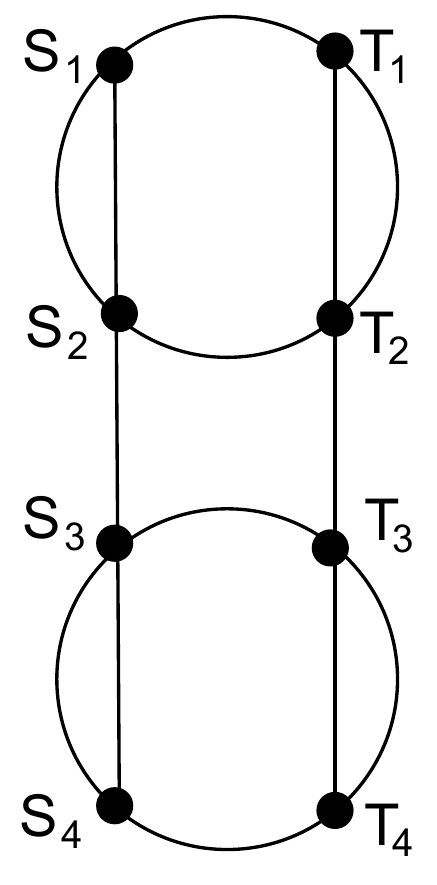}
\caption{A module and its four sides.}
\label{module}
\end{figure}

Let
\[
\Lambda=\{A,B,\ldots,R,U,V,W,X,Y,Z\}
\]
be the set of the twenty-four letters of the alphabet other than $S$ and $T$.
We use the following six column sides:
\[
\begin{aligned}
L_1&=\{A,B,C,U\}, & R_1&=\{J,K,L,X\},\\
L_2&=\{D,E,F,V\}, & R_2&=\{M,N,O,Y\},\\
L_3&=\{G,H,I,W\}, & R_3&=\{P,Q,R,Z\}.
\end{aligned}
\]
For $1\le i,j\le3$, the module $m_{ij}$ has left side $L_i$, right side
$R_j$, and the two loop sides listed below:
\[
\begin{array}{c|cc}
 & \text{top loop} & \text{bottom loop}\\ \hline
m_{11}&\{A,J,K,B\}&\{C,L,X,U\}\\
m_{12}&\{A,M,N,C\}&\{B,O,Y,U\}\\
m_{13}&\{B,P,Q,C\}&\{A,R,Z,U\}\\
m_{21}&\{D,J,L,E\}&\{F,K,X,V\}\\
m_{22}&\{D,M,O,F\}&\{E,N,Y,V\}\\
m_{23}&\{E,P,R,F\}&\{D,Q,Z,V\}\\
m_{31}&\{G,K,L,H\}&\{I,J,X,W\}\\
m_{32}&\{G,N,O,I\}&\{H,M,Y,W\}\\
m_{33}&\{H,Q,R,I\}&\{G,P,Z,W\}.
\end{array}
\tag{\ensuremath{\Sigma}}
\label{eq:sigma-sides}
\]
The set of all twenty-four sides---the six columns and the eighteen
loops---is denoted by $\mathcal S$.

\begin{figure}[p]
\centering
\includegraphics[scale=0.46]{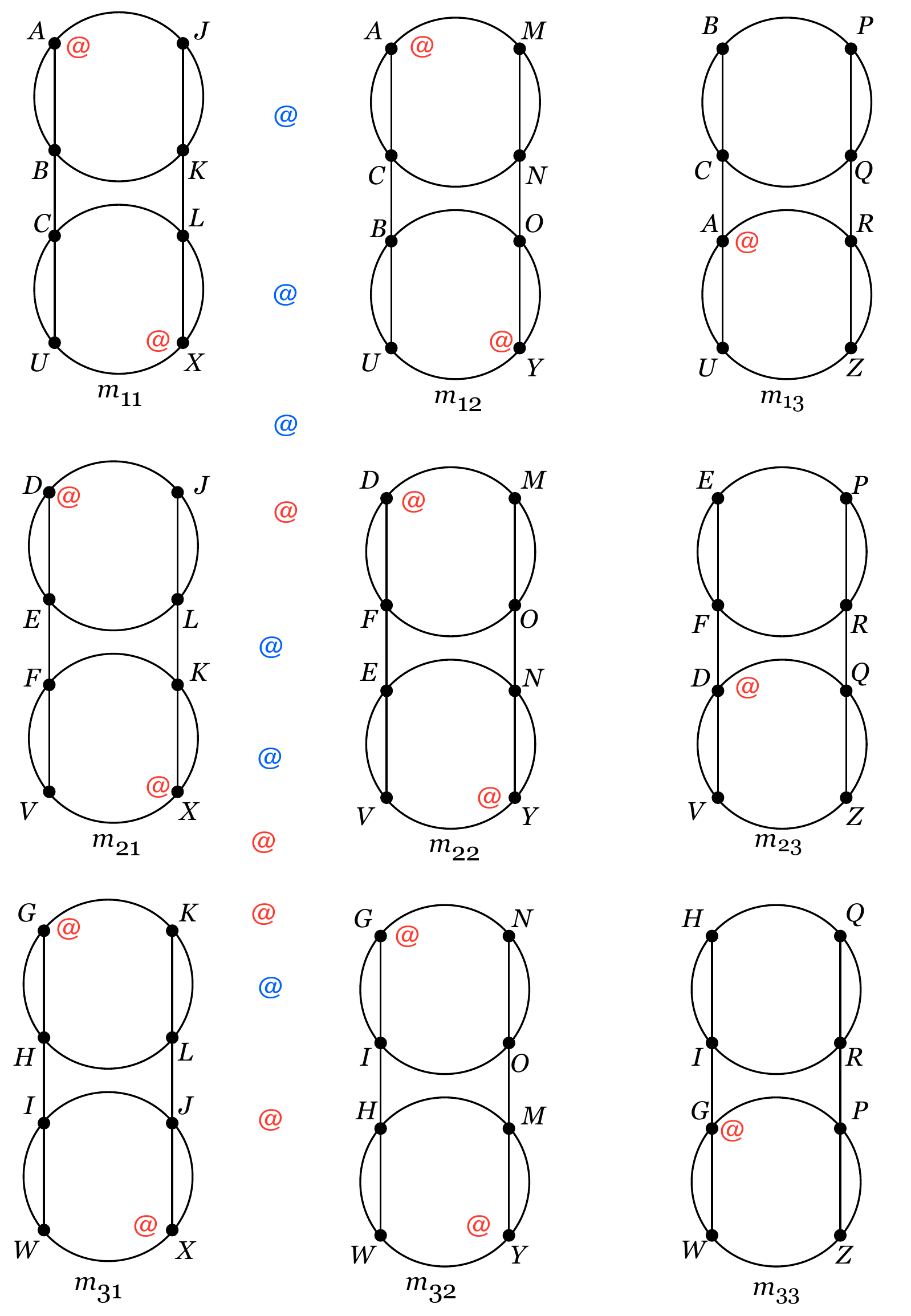}
\caption{The configuration $\Sigma$, a $3\times3$ array of modules $m_{ij}$.}
\label{sigma}
\end{figure}

A \emph{word} is a subset of a side $x\in\mathcal S$.  If $w\subseteq x$,
its complement on $x$ is
\[
w^x=x\setminus w.
\]
The same set of letters may be a word on more than one side, and its complement
may then depend on the side.

\begin{definition}
\label{def:cong}
Two words $u$ and $v$ are \emph{congruent}, written $u\approx v$, if there
exist sides $x,y\in\mathcal S$ such that $u\subseteq x$, $v\subseteq y$, and
\[
u^x=v^y.
\]
The congruence class of $w$ is denoted by $[[w]]$ and is called the
\emph{vertex} represented by $w$.
\end{definition}

The following finite incidence lemma records all of the combinatorial facts
needed below.  Its proof is included because it is precisely these facts that
ensure that the quotient operations are well defined.

\

\begin{lemma}[Incidence lemma]
\label{lem:sigma-incidence}
For the twenty-four sides in~\eqref{eq:sigma-sides}, the following assertions
hold.
\begin{enumerate}
\item Every letter of $\Lambda$ lies on exactly four sides.
\item Distinct sides meet in zero, one, or two letters.
\item Every one-letter word is congruent only to itself.
\item For each letter $a$, the four three-letter complements of $\{a\}$ on
      the four sides containing $a$ are mutually congruent, and no other
      three-letter word belongs to that class.
\item Among the distinct two-letter words, the non-singleton congruence
      classes are exactly the following eighteen pairs:
      \[
      \begin{array}{lll}
      \{A,B\}\approx\{L,X\}, &
      \{A,C\}\approx\{O,Y\}, &
      \{A,U\}\approx\{P,Q\},\\
      \{B,C\}\approx\{R,Z\}, &
      \{B,U\}\approx\{M,N\}, &
      \{C,U\}\approx\{J,K\},\\
      \{D,E\}\approx\{K,X\}, &
      \{D,F\}\approx\{N,Y\}, &
      \{D,V\}\approx\{P,R\},\\
      \{E,F\}\approx\{Q,Z\}, &
      \{E,V\}\approx\{M,O\}, &
      \{F,V\}\approx\{J,L\},\\
      \{G,H\}\approx\{J,X\}, &
      \{G,I\}\approx\{M,Y\}, &
      \{G,W\}\approx\{Q,R\},\\
      \{H,I\}\approx\{P,Z\}, &
      \{H,W\}\approx\{N,O\}, &
      \{I,W\}\approx\{K,L\}.
      \end{array}
      \]
      The remaining seventy-two two-letter words are congruent only to
      themselves.  Hence no congruence class of two-letter words has more than
      two members.
\item If two distinct sides $x$ and $y$ are mapped to their sets of vertices
      \[
      \mathcal B_x=\{[[w]]:w\subseteq x\},\qquad
      \mathcal B_y=\{[[w]]:w\subseteq y\},
      \]
      then $|\mathcal B_x\cap\mathcal B_y|$ is $2$, $4$, or $8$, and the
      intersection is closed under the Boolean operations induced from either
      side.
\item Every family of pairwise compatible vertices is contained in
      $\mathcal B_x$ for some side $x\in\mathcal S$.
\end{enumerate}
\end{lemma}

\begin{proof}
All assertions are finite consequences of the displayed list of sides.  We give
an exhaustive verification by word size.

Each of $A,\ldots,R$ occurs on its column and on three loops; each of
$U,\ldots,Z$ also occurs on its column and on three loops.  This proves (1).
Comparing the twenty-four displayed four-element sets proves (2).

A one-letter word $\{a\}$ can be congruent to $\{b\}$ only if complements of
these two words are equal as sets.  Equality of those three-element complements
forces $a=b$, proving (3).  By (1), each $a$ has four three-letter complements;
they are congruent by definition.  Conversely, equality of the complements of
two three-letter words identifies their common missing letter, proving (4).

For two-letter words, the displayed list in item~(5) gives all
non-singleton classes directly.  Each pair has a common two-letter complement
on the two indicated sides.  Conversely, equality of two-letter complements
forces one of the eighteen displayed coincidences when the side list
in~\eqref{eq:sigma-sides} is inspected.  Since the twenty-four sides contain
$24\binom{4}{2}=144$ two-letter occurrences and the side intersections identify
exactly thirty-six duplicate occurrences, there are $108$ distinct two-letter
words.  The eighteen two-element classes account for thirty-six of them, leaving
seventy-two singleton classes.  This proves (5).

For (6), let $x\ne y$.  By item~(2), $|x\cap y|$ is $0$, $1$, or $2$.  The
vertices represented on both sides are precisely the Boolean subalgebra
generated by the singleton vertices belonging to $x\cap y$, together with any
identified complementary blocks forced by the congruences in item~(5).
Direct substitution from the eighteen-pair table shows that the resulting
intersection has respectively $2$, $4$, or $8$ elements and is closed under
complement, meet, and join on either side.

For completeness, assertion~(7), as well as the preceding finite counts, is
also certified by the exhaustive verification described in
Appendix~\ref{app:v140-certificate}.  That verification constructs the
compatibility graph from the twenty-four displayed sides and enumerates all
maximal cliques; the output consists of exactly the twenty-four sets
$\mathcal B_x$, each of cardinality sixteen.  Hence every pairwise compatible
family lies in one of them, proving (7).
\end{proof}

\subsection{Congruence classes and quotient operations}

\begin{proposition}
\label{prop:cong-equivalence}
The relation $\approx$ is an equivalence relation on the set of all words.
Its equivalence classes have the following profile:
\begin{center}
\small
\begin{tabular}{c|c|p{5.0cm}|c}
\textbf{Word size} & \textbf{Distinct words} & \textbf{Class profile} &
\textbf{Classes}\\ \hline
$0$ & $1$ & one class of size $1$ & $1$\\
$1$ & $24$ & twenty-four classes of size $1$ & $24$\\
$2$ & $108$ & eighteen classes of size $2$ and seventy-two classes of size $1$ & $90$\\
$3$ & $96$ & twenty-four classes of size $4$ & $24$\\
$4$ & $24$ & one class of size $24$ & $1$
\end{tabular}
\end{center}
Consequently there are
\[
1+24+90+24+1=140
\]
vertices.
\end{proposition}

\begin{proof}
Reflexivity and symmetry follow immediately from Definition~\ref{def:cong}.
The class description in Lemma~\ref{lem:sigma-incidence}(3)--(5), together with
the evident single classes of the empty word and all four-letter words, partitions
the set of words into pairwise disjoint blocks.  Congruence is equality within
these blocks, and is therefore transitive.  The table and the count follow.
\end{proof}

Let $\mathcal V_{140}$ denote the set of congruence classes.  A vertex is said
to lie on side $x$ if it has a representative contained in $x$; two vertices
are \emph{compatible} if they lie on a common side.  For compatible vertices,
choose representatives $u,v\subseteq x$ on a common side and set
\begin{equation}
\label{eq:v140-operations}
[[u]]\sqcap[[v]]=[[u\cap v]],\qquad
[[u]]\sqcup[[v]]=[[u\cup v]],\qquad
[[u]]^{\perp}=[[u^x]].
\end{equation}
Also define
\[
[[u]]\le [[v]]
\quad\Longleftrightarrow\quad
\text{there are representatives }u',v'\subseteq x
\text{ on a common side with }u'\subseteq v'.
\]

\begin{proposition}[Well-definedness]
\label{prop:v140-welldefined}
The operations and order in~\eqref{eq:v140-operations} are independent of the
chosen representatives and of the chosen common side.
For every $x\in\mathcal S$, the map
\[
q_x:\mathcal P(x)\longrightarrow\mathcal B_x,
\qquad w\longmapsto[[w]],
\]
is a Boolean-algebra isomorphism.
\end{proposition}

\begin{proof}
First, if $u,v\subseteq x$ and $[[u]]=[[v]]$, then $u=v$.  Congruence
preserves word size, so $|u|=|v|$, and it suffices to inspect
Lemma~\ref{lem:sigma-incidence}(3)--(5) size by size.  Words of size $0$, $1$
and $3$ that are congruent and lie on the same side coincide: for size~$1$ by
item~(3), and for size~$3$ because a three-letter word lies on only one side,
sides meeting in at most two letters by item~(2).  Each side contains exactly
one word of size~$4$.  For size~$2$, item~(5) lists the eighteen non-singleton
classes explicitly, and in each of them the two members have no common side,
since their union is a four-element set that does not appear
in~\eqref{eq:sigma-sides}.  Hence $q_x$ is injective,
and it is surjective by the definition of $\mathcal B_x$.  It therefore
transports the Boolean operations of $\mathcal P(x)$ to $\mathcal B_x$.

Suppose a vertex, or a compatible pair of vertices, is represented on two sides
$x$ and $y$.  By Lemma~\ref{lem:sigma-incidence}(6),
$\mathcal B_x\cap\mathcal B_y$ is a Boolean subalgebra of both
$\mathcal B_x$ and $\mathcal B_y$.  Complement, meet, join, and order computed
in either side therefore agree on the intersection.  This proves independence
of all choices.
\end{proof}

As the component family we take
\begin{equation}
\label{eq:pi140}
\Pi_{140}
=\{\,\mathcal B:\mathcal B\text{ is a Boolean subalgebra of }
        \mathcal B_x\text{ for some }x\in\mathcal S\,\},
\end{equation}
in exact analogy with $\Pi^\circ$ in~\eqref{eq:p12}.  Including the
subalgebras, and not merely the twenty-four algebras $\mathcal B_x$, is what
makes~\rm(P3) available.

\begin{theorem}
\label{thm:v140-pba}
With compatibility, the component family~\eqref{eq:pi140}, the partial
operations in~\eqref{eq:v140-operations}, and the order above,
$(\mathcal V_{140},\Pi_{140})$ is a partial Boolean algebra.  Its maximal
components are precisely the twenty-four Boolean algebras
$\mathcal B_x$ with $x\in\mathcal S$.
\end{theorem}

\begin{proof}
Every vertex has a representative on some side, so it belongs to a component.
If $\mathcal B\subseteq\mathcal B_x$ and $\mathcal B'\subseteq\mathcal B_y$ are
components, then $\mathcal B\cap\mathcal B'$ is contained in
$\mathcal B_x\cap\mathcal B_y$, which by Lemma~\ref{lem:sigma-incidence}(6) is a
Boolean subalgebra of both $\mathcal B_x$ and $\mathcal B_y$; hence
$\mathcal B\cap\mathcal B'$ is an intersection of two Boolean subalgebras of
$\mathcal B_x$ and is therefore itself a Boolean subalgebra of
$\mathcal B_x$, so it lies in $\Pi_{140}$.  This gives~\rm(P3),
and~\rm(P4) is immediate from the definition of $\Pi_{140}$.
By Proposition~\ref{prop:v140-welldefined}, each $\mathcal B_x$ is a Boolean
algebra, all such algebras share the common zero $[[\emptyset]]$ and common unit
(the class of all four-letter sides), and orthocomplementation is globally
well defined.  Lemma~\ref{lem:sigma-incidence}(6) shows that intersections of
components are Boolean subalgebras.  Lemma~\ref{lem:sigma-incidence}(7) shows
that every finite pairwise compatible family lies in a common component.  The
partial order is the Boolean order on each component; well-definedness and
transitivity follow from Proposition~\ref{prop:v140-welldefined} and the common-
component property.  These are exactly the partial-Boolean-algebra axioms stated
in Section~\ref{sec:pba}.  Maximality follows from
Lemma~\ref{lem:sigma-incidence}(7), since the maximal compatible sets are the
$\mathcal B_x$.
\end{proof}

\begin{corollary}
\label{cor:v140-dimension}
The partial Boolean algebra $\mathcal V_{140}$ is $4$-dimensional.  Its atoms are
the twenty-four singleton vertices $[[\{a\}]]$ ($a\in\Lambda$), and its frames
are exactly the twenty-four four-element sets
\[
F_x=\{[[\{a\}]]:a\in x\},\qquad x\in\mathcal S.
\]
\end{corollary}

\begin{proof}
Under the isomorphism $q_x$, the atoms of $\mathcal B_x$ are the singleton
subsets of $x$.  Lemma~\ref{lem:sigma-incidence}(3) shows that singleton words
remain distinct in the quotient, so these are precisely the atoms of
$\mathcal V_{140}$.  Each maximal component has four atoms, and by
Theorem~\ref{thm:v140-pba} the maximal components are exactly the
$\mathcal B_x$.  The conclusion follows from Definition~\ref{def:finite-dimensional}.
\end{proof}

\subsection{The parity obstruction}

A KS-function on $\mathcal V_{140}$ is therefore exactly a function
$f:\Lambda\to\{0,1\}$ assigning the value $1$ to exactly one letter of every
side in $\mathcal S$.

\begin{theorem}
\label{thm:noKS}
The partial Boolean algebra $\mathcal V_{140}$ does not have the KS-property.
\end{theorem}

\begin{proof}
Suppose that such a function $f$ exists.  Consider the nine top loops
\[
\mathcal L=\bigl\{
\begin{aligned}[t]
&\{A,J,K,B\},\ \{A,M,N,C\},\ \{B,P,Q,C\},\\
&\{D,J,L,E\},\ \{D,M,O,F\},\ \{E,P,R,F\},\\
&\{G,K,L,H\},\ \{G,N,O,I\},\ \{H,Q,R,I\}
\end{aligned}
\bigr\}.
\]
Since each member of $\mathcal L$ is a side, exactly one of its four letters has
value $1$.  Hence
\begin{equation}
\label{eq:parity-odd}
\sum_{\ell\in\mathcal L}\ \sum_{a\in\ell}f(a)=9.
\end{equation}
On the other hand, each of $A,B,\ldots,R$ occurs in exactly two of these nine
loops, while none of $U,V,W,X,Y,Z$ occurs in them.  Therefore the left-hand side
of~\eqref{eq:parity-odd} equals
\[
2\sum_{a\in\{A,B,\ldots,R\}}f(a),
\]
which is even.  This contradicts~\eqref{eq:parity-odd}.  Thus no KS-function
exists.
\end{proof}

\begin{corollary}
\label{cor:macrouncertain}
There is a consistent theory, formulated in classical propositional logic and
having core $\mathcal V_{140}$, with no dispersion-free model.  In every model
of that theory at least one observable fails to be sharp---that is, it is either
non-actualizable or has uncertainty at least~$1$---and no global hidden-value
assignment removes all such failures.
\end{corollary}

\begin{proof}
Choose a set $U$ and a surjection $g:U\twoheadrightarrow\mathcal V_{140}$.
The induced theory $\mathcal T_g$ is consistent by Theorem~\ref{thmC}, and its
core is isomorphic to $\mathcal V_{140}$ by Theorem~\ref{thm:corepba}.
By Theorem~\ref{thm:noKS} the algebra $\mathcal V_{140}$ lacks the
KS-property, so Theorem~\ref{thm:dispfree} shows that $\mathcal T_g$ has no
dispersion-free model and that $\mathcal V_{140}$ admits no two-valued
homomorphism.  If some model made every finite-spectrum observable sharp, then
Theorem~\ref{thm:sharp-implies-ks} would produce a KS-function, contradicting
Theorem~\ref{thm:noKS}.  Hence in every model some observable is not sharp.
\end{proof}

\begin{remark}
\label{rem:macro-careful}
Two refinements of the preceding statement are worth recording, because the
weaker formulation ``every model misses some pre-frame, hence some observable
is uncertain'' is not a valid inference.  First, missing a pre-frame does
\emph{not} by itself imply non-sharpness: as noted after
Proposition~\ref{prop:disp-sharp}, $\mathcal V_{12}$ has models that miss both
pre-frames and yet make both observables sharp.  Theorem~\ref{thm:sharp-implies-ks}
is what closes this gap.  Second, ``not sharp'' is genuinely weaker than
``uncertainty at least~$1$'', because the pointer set may be empty
(Remark~\ref{rem:empty-pointer}).  Both alternatives occur in
$\mathcal V_{140}$: the set of all vertices represented by three-letter and
four-letter words is an upward-closed cluster, hence a model, and in it every
one of the $24$ observables has empty pointer set, since the four three-letter
subsets of a side have empty intersection.
\end{remark}
This is the second of the two examples promised in the abstract.
Whereas the uncertainties of the $\mathcal V_{12}$ world are mere ignorance of a
hidden classical state, those of the $\mathcal V_{140}$ world are intrinsic: no
consistent refinement removes them.

\begin{remark}[Hilbert-space realization]
\label{rem:peres}
Let $\mathcal R\subseteq\mathbb R^{4}$ consist of the twenty-four rays spanned
by the vectors
\[
(1,0,0,0)\ \text{and its permutations},\qquad
(1,\pm1,0,0)\ \text{and its permutations},\qquad
(1,\pm1,\pm1,\pm1),
\]
the Peres configuration \cite{Peres1991}.  These rays fall into exactly
twenty-four orthogonal bases, each ray belonging to exactly four of them, and
two distinct bases share $0$, $1$ or $2$ rays.  The bipartite incidence
structure of rays and bases is isomorphic to that of $\Lambda$ and
$\mathcal S$ in~\eqref{eq:sigma-sides}; an isomorphism is produced by the
verification script of Appendix~\ref{app:v140-certificate}.  Under any such
isomorphism the sub-partial-Boolean-algebra of $\mathcal V(\mathbb R^{4})$
generated by $\mathcal R$ and its bases is isomorphic to
$(\mathcal V_{140},\Pi_{140})$, and Theorem~\ref{thm:noKS} becomes the
Kochen--Specker theorem for the Peres set.  Consequently the macroscopic
reading of Corollary~\ref{cor:macrouncertain} is not an artefact of the
combinatorics: the same core is realized by genuine quantum propositions in
dimension four.
\end{remark}

\subsection{Interpretive scope: macroscopic realism and contextuality}
\label{subsec:interpretive-scope}

The preceding result is a theorem about global truth assignments in a classical
propositional theory whose core is a partial Boolean algebra.  It has a natural
formal resemblance to the assumptions used in discussions of macroscopic
realism: a dispersion-free model assigns one definite value in every
measurement frame, independently of which frame is considered.  In this
limited sense, Theorem~\ref{thm:dispfree} identifies KS-colourability with the
existence of a global definite-value assignment, while
Theorem~\ref{thm:noKS} exhibits an explicit obstruction to such an assignment.

This resemblance should not be overstated.  The present construction does not
by itself constitute a Leggett--Garg experiment or a derivation of a
Leggett--Garg inequality \cite{LeggettGarg1985}.  Such inequalities concern
probabilities and temporal correlations obtained from specified measurement
procedures.  The framework developed here presently supplies a space of
models, a structural notion of incompatibility, and a proposed update rule; it
does not yet supply probabilities for the possible outcomes of that update.
Consequently, no numerical Leggett--Garg prediction follows from the results
proved in this paper.

The obstruction in Theorem~\ref{thm:noKS} is also structurally analogous to
contextuality: there is no single two-valued assignment whose restriction to
every component is a Boolean homomorphism.  Contextuality has been studied
through graph-theoretic and resource-theoretic methods
\cite{Cabello2014,Howard2014}, and in particular operational forms of
contextuality are known to be relevant to quantum advantage in specific
computational models.  The contribution made here is more modest and more
logical: the same compatibility structure can be represented as the core of a
consistent classical theory, and the existence or nonexistence of a global
valuation becomes a statement about the models of that theory.  No claim of a
computational speed-up, a resource monotone, or an operational implementation
is made.

\section{Outlook}
\label{sec:outlook}

The results suggest several concrete continuations.

First, one may supplement the measurement-update rule of
Section~\ref{sec:dynamics} with a probability law on its admissible outcomes.
Only after such a law has been specified can temporal correlations be defined
and Leggett--Garg-type inequalities be meaningfully investigated.  The finite
examples $\mathcal V_{12}$ and $\mathcal V_{140}$ provide test cases for such a
study, but the present results do not determine in advance whether either
example satisfies or violates a particular inequality.

Second, the update rule defines a finite transition structure on the model
space whenever the underlying partial Boolean algebra and the chosen set of
initial formul\ae\ are finite.  Its reachable states, fixed points, recurrent
classes, and the order dependence of successive measurements can therefore be
studied directly.  This would clarify which features of the proposed dynamics
follow solely from the compatibility structure and which depend on an added
probability law.

Third, the quantitative notion of uncertainty introduced in
Section~\ref{fuzzy} may be compared across models and observables.  To turn it
into a contextuality measure would require specifying a class of free
operations and proving monotonicity under those operations.  Neither step is
carried out here, so this remains an open problem rather than a consequence of
the present theory.

Finally, the order-theoretic structure of the family of models of
$\mathcal T_g$ deserves separate investigation.  In particular, it is natural
to ask which closure or lattice properties this family possesses and how those
properties are related to the partial Boolean algebra in the core.  This
question may help make precise the distinction, emphasized above, between
states, models, and global value assignments.

\appendix
\section{Appendix: relation to quantum states}
\label{App:QT}

This appendix justifies the terminology of Section~\ref{pm} and, equally
importantly, states its limitation.
A model of $\mathcal T_g$ is a two-valued valuation on the propositional
language, whereas a quantum state assigns probabilities in $[0,1]$.
The model therefore cannot encode the complete probability law of a quantum
state.
What it can encode exactly is the state's \emph{certainty set}: the collection
of quantum propositions having probability~$1$.

Let $\mathcal H$ be a finite-dimensional complex Hilbert space of dimension at
least~$3$, and let $\mathcal V_{\mathcal H}$ be its partial Boolean algebra of
closed subspaces.
By Gleason's theorem, every countably additive probability measure on the
projections of $\mathcal H$ is represented by a density operator $\rho$ through
\[
\mu_\rho(S)=\operatorname{tr}(\rho J_S),
\]
where $J_S$ denotes the orthogonal projection onto the subspace $S$
\cite{Gleason1957}.
A state is pure when $\rho$ has rank one and mixed otherwise.

\subsection{Pure states}
\label{subsec:pure}

Let $\rho=J_\alpha$ be the rank-one projection onto a ray $\alpha$.
For every subspace $S$,
\[
\mu_\rho(S)=\operatorname{tr}(J_\alpha J_S)=1
\quad\Longleftrightarrow\quad
\alpha\le S.
\]
Now let $g:U\twoheadrightarrow\mathcal V_{\mathcal H}$ be a canonical function
and choose $e\in U$ with $g(e)=\alpha$.
Then
\[
\mathfrak M_e=\{e\}^*
=\{q\in U:g(e)\le g(q)\}
\]
is a model, and
\[
q\in\mathfrak M_e
\quad\Longleftrightarrow\quad
\mu_\rho(g(q))=1.
\]
Thus the pure-state certainty model records exactly the propositions certain in
the pure state $\rho$.
Truth value~$0$ means only that the probability is strictly less than~$1$; it
does not distinguish probability~$0$ from intermediate probabilities.

\subsection{Mixed states and support projections}
\label{subsec:mixed}

Let $\rho$ be a density operator with spectral decomposition
\begin{equation}
\label{eq:spectral}
\rho=\sum_{i=1}^k x_iJ_i,
\qquad x_i>0,
\qquad \sum_{i=1}^k x_i\dim(\operatorname{ran}J_i)=1,
\end{equation}
where the $J_i$ are the pairwise orthogonal spectral projections associated with
the nonzero eigenvalues.
Its support projection is
\[
J_R=\sum_{i=1}^k J_i,
\]
and $R=\operatorname{ran}J_R$ is the support subspace of $\rho$.

\begin{proposition}
\label{prop:mix}
For every subspace $S\le\mathcal H$,
\[
\operatorname{tr}(\rho J_S)=1
\quad\Longleftrightarrow\quad
R\le S.
\]
Consequently, if $R\not\le S$, then
$0\le\operatorname{tr}(\rho J_S)<1$.
\end{proposition}

\begin{proof}
Since $\rho=J_R\rho J_R$, the condition $R\le S$ implies
$J_SJ_R=J_R$, hence $\rho J_S=\rho$ and
$\operatorname{tr}(\rho J_S)=1$.
Conversely, suppose $\operatorname{tr}(\rho J_S)=1$.
Then
\[
0=\operatorname{tr}(\rho)-\operatorname{tr}(\rho J_S)
  =\operatorname{tr}\bigl(\rho^{1/2}(I-J_S)\rho^{1/2}\bigr).
\]
The operator inside the trace is positive, so it must vanish.
Therefore $(I-J_S)\rho^{1/2}=0$, and the range of $\rho^{1/2}$, which is $R$,
is contained in $S$.
The final assertion follows from positivity and normalization.
\end{proof}

\begin{proposition}
\label{prop:everysubspace}
Every nonzero subspace $R\le\mathcal H$ is the support of a density operator.
\end{proposition}

\begin{proof}
Let $e_1,\dots,e_m$ be an orthonormal basis of $R$ and choose numbers
$a_i>0$ with $\sum_i a_i=1$.
Then
\[
\rho=\sum_{i=1}^m a_i|e_i\rangle\langle e_i|
\]
is positive, has trace~$1$, and has support exactly $R$.
\end{proof}

Density operators with the same support have the same probability-one
propositions, although they may assign different probabilities below~$1$.
Hence support subspaces classify density operators only up to equality of their
certainty sets.
This motivates the following abstract terminology.

\begin{definition}
\label{def:absmix}
Let $(\mathcal V,\Pi)$ be an atomic $n$-dimensional partial Boolean algebra.
A nonzero non-atomic element $r\in\mathcal V$ is called a \emph{support-state
element}.
If $g(a)=r$, the model
\[
\mathfrak M_a=\{a\}^*=\{q\in U:r\le g(q)\}
\]
is called the \emph{support-state certainty model} based on $a$.
\end{definition}

In the Hilbert-space case, every support-state element is the support of at least
one mixed density operator by Proposition~\ref{prop:everysubspace}, and
Proposition~\ref{prop:mix} shows that $\mathfrak M_a$ records exactly its
probability-one propositions.
No claim is made that an abstract support-state element carries the convex or
probabilistic structure of a density operator; the correspondence concerns only
the induced certainty set.

\section{The partial Boolean algebra of a single module}
\label{app:single-module}

This optional appendix records the quotient associated with one module.  It is
not used in the proofs of the main results, but it gives a smaller illustration
of the construction underlying $\mathcal V_{140}$.  Note that the resulting
partial Boolean algebra has $36$ vertices; the count is carried out in full
below, because the number is easy to get wrong.

Let
\[
X=\{S_1,S_2,S_3,S_4,T_1,T_2,T_3,T_4\}
\]
and let the four sides be
\[
\begin{aligned}
l&=\{S_1,S_2,S_3,S_4\},
& r&=\{T_1,T_2,T_3,T_4\},\\
k_1&=\{S_1,S_2,T_1,T_2\},
& k_2&=\{S_3,S_4,T_3,T_4\}.
\end{aligned}
\]
These are the four sides of the module in Figure~\ref{module}.  A \emph{word}
is a subset of one of these sides.  If $w\subseteq x$, write
$w^x=x\setminus w$.  For words $u$ and $v$, define
\[
u\approx_0 v
\quad\Longleftrightarrow\quad
u^x=v^y
\quad\text{for some sides }x\supseteq u\text{ and }y\supseteq v.
\]
Write $[[w]]_0$ for the equivalence class of $w$.

\begin{proposition}
\label{prop:single-module-count}
The relation $\approx_0$ is an equivalence relation, and it has exactly
$36$ equivalence classes.  Their profile is
\[
\begin{array}{c|c|c|c}
\text{word size}&\text{distinct words}&\text{class profile}&\text{classes}\\ \hline
0&1&1\text{ singleton}&1\\
1&8&8\text{ singletons}&8\\
2&20&2\text{ classes of size }2+16\text{ singletons}&18\\
3&16&8\text{ classes of size }2&8\\
4&4&1\text{ class of size }4&1
\end{array}
\]
Consequently,
\[
1+8+18+8+1=36.
\]
\end{proposition}

\begin{proof}
The empty word forms a singleton class, and all four sides are congruent because
their complements are empty.  Every one-letter word is congruent only to
itself.  Each letter lies on exactly two sides, so the two three-letter
complements of that letter form one class; hence there are eight classes of
three-letter words.

There are twenty distinct two-letter words.  The only non-singleton classes are
\[
\{S_1,S_2\}\approx_0\{T_3,T_4\},
\qquad
\{T_1,T_2\}\approx_0\{S_3,S_4\},
\]
because the respective common complements are $\{S_3,S_4\}$ and
$\{S_1,S_2\}$.  The remaining sixteen two-letter words are singleton classes.
These disjoint classes exhaust all words, so they form a partition; equivalently,
$\approx_0$ is an equivalence relation.  The displayed count follows.
\end{proof}

Let $\mathcal V_{36}$ denote the set of these classes.  For a side $x$, put
\[
\mathcal B_x=\{[[w]]_0:w\subseteq x\}.
\]
If two vertices have representatives $u,v\subseteq x$ on a common side, define
\[
[[u]]_0\sqcap[[v]]_0=[[u\cap v]]_0,
\qquad
[[u]]_0\sqcup[[v]]_0=[[u\cup v]]_0,
\qquad
[[u]]_0^\perp=[[u^x]]_0.
\]

\begin{theorem}
\label{thm:single-module-pba}
The operations above are independent of the chosen common side and of the
chosen representatives.  With components consisting of
$\mathcal B_l,\mathcal B_r,\mathcal B_{k_1},\mathcal B_{k_2}$ and all their
Boolean subalgebras, $\mathcal V_{36}$ is a four-dimensional partial Boolean
algebra.  Its four maximal components are precisely the four displayed
$16$-element Boolean algebras.
\end{theorem}

\begin{proof}
The class enumeration in Proposition~\ref{prop:single-module-count} permits a
direct check.  If a class has representatives on two sides, taking the
side-complement of either representative produces the same class.  The same
enumeration shows that unions and intersections of representatives on a common
side depend only on their classes.  Thus each map
\[
\mathcal P(x)\longrightarrow\mathcal B_x,
\qquad w\longmapsto[[w]]_0,
\]
is a Boolean-algebra isomorphism.

The pairwise intersections of the four maximal components have respectively
$4$ or $8$ elements and are Boolean subalgebras of both components.  A direct
inspection of the four sides shows that every pairwise compatible family is
contained in one of the $\mathcal B_x$.  Hence the partial Boolean algebra
axioms hold.

The atoms are the eight singleton classes $[[\{a\}]]_0$, $a\in X$.  The atoms
on each side form a frame of four elements, and the four maximal components are
exactly the Boolean algebras generated by these four frames.  Therefore the
partial Boolean algebra is four-dimensional.
\end{proof}

\begin{remark}
The natural miscount here gives $28$: it arises from recording only eight
singleton classes of two-letter words.  There are in fact sixteen such singleton
classes; with the two non-singleton two-letter classes, the number of classes of
two-letter words is $18$, and the total is $36$.
\end{remark}

\section{Finite verification certificate for $\mathcal V_{140}$}
\label{app:v140-certificate}

The construction in Section~\ref{sec:4dim} is finite, and every incidence and
well-definedness assertion used there can be checked exhaustively from the
side list~\eqref{eq:sigma-sides}.  This appendix records the verification
procedure so that the finite part of the argument is reproducible independently
of the geometric drawing.

Represent each side by a four-element set and form the set $W$ of all subsets
of all sides.  For $u,v\in W$, declare
\[
u\approx v
\quad\Longleftrightarrow\quad
x\setminus u=y\setminus v
\quad\text{for some sides }x\supseteq u,
\ y\supseteq v.
\]
The verification proceeds as follows.

\begin{enumerate}
\item Enumerate the $24$ sides and all of their subsets.  After duplicate words
      are removed, one obtains $253$ distinct words.
\item Evaluate the relation $\approx$ on all ordered pairs in $W\times W$ and
      check reflexivity, symmetry, and transitivity directly.
\item Enumerate the resulting equivalence classes and sort them by word size.
      The output is
      \[
      1,\ 24,\ 90,\ 24,\ 1,
      \]
      with the two-letter classes split into eighteen classes of size two and
      seventy-two classes of size one.  Hence there are $140$ classes.
\item For every side $x$ and every $u,v\subseteq x$, record the class identifiers
      of $x\setminus u$, $u\cup v$, and $u\cap v$.  Group these records by the
      class identifiers of $u$ and $v$.  Each group contains a single output
      value.  This proves representative-independence of complement, join, and
      meet.
\item For each pair of side components, check closure of their intersection
      under the recorded complement, join, and meet operations.
\item Construct the compatibility graph on the $140$ classes, joining two
      classes when they have representatives on a common side.  Exhaustive
      maximal-clique enumeration returns exactly $24$ maximal cliques, each of
      size $16$, and these are precisely the side components $\mathcal B_x$.
\item For the nine top loops used in Theorem~\ref{thm:noKS}, count letter
      incidences.  Each of $A,\ldots,R$ occurs twice and each of
      $U,\ldots,Z$ occurs zero times, certifying the parity calculation.
\item Independently, perform an exhaustive search over the $2^{24}$ functions
      $f:\Lambda\to\{0,1\}$ (pruned by unit propagation on the sides) and
      confirm that none assigns the value $1$ to exactly one letter of every
      side.  This certifies Theorem~\ref{thm:noKS} without appeal to the parity
      argument.
\item Construct the twenty-four Peres rays of Remark~\ref{rem:peres}, form
      their orthogonality graph, enumerate its maximal cliques, and test the
      resulting point--block incidence structure for isomorphism with
      $(\Lambda,\mathcal S)$.  The test succeeds, certifying the Hilbert-space
      realization.
\end{enumerate}

The authors have carried out each of these steps in a short Python program,
which confirms every assertion listed above: the $24$ sides and $253$ distinct
words; the class profile $1+24+90+24+1=140$; the fact that $\approx$ is an
equivalence relation; representative-independence of complement, join, and meet;
closure of component intersections; the $24$ maximal compatible components, each
a Boolean algebra of size $16$; the $24$ four-element frames, one per side; the
parity incidences of Theorem~\ref{thm:noKS}; the nonexistence of a KS-function
by exhaustive search; and the incidence isomorphism with the Peres $24$-ray
configuration of Remark~\ref{rem:peres}.  The same program reproduces the counts
quoted for $\mathcal T_{12}$ in Remark~\ref{rem:empty-pointer} and after
Proposition~\ref{prop:disp-sharp}, namely $53$ models, of which $5$ are
dispersion-free, $18$ make both $P_0$ and $Q_0$ sharp, and $11$ leave one of them
with an empty pointer set.  The code is available from the authors on request.

The procedure is entirely deterministic: no random choices, numerical
tolerances, floating-point arithmetic, or external data enter it, and every
step above can equally be carried out by hand from the side
list~\eqref{eq:sigma-sides}.

\section*{Author information}

\noindent{\bf$^*$Othman Q. Malhas (1942--2024)}\\
Othman Q. Malhas earned his Ph.D. from the University of Sussex,
England, and was a professor of mathematics at Yarmouk University,
Irbid, Jordan.

\medskip

\noindent{\bf$^\dag$Bacim Alali}\\
Professor, Department of Mathematics, Kansas State University\\
Manhattan, Kansas, USA\\
Email: \texttt{bacimalali@math.ksu.edu}

\section*{Declarations}

\noindent\textbf{Author contributions.}
Othman Q. Malhas developed the original conception and mathematical
framework of this work. Bacim Alali carried out the mathematical
validation, substantial revision, completion, and preparation of the
manuscript for publication.

\medskip

\noindent\textbf{Completion of the manuscript.}
Following the death of Othman Q. Malhas in 2025, his wife, Enaya I. Mango,
informed Bacim Alali, a former student of Professor Malhas, that Professor
Malhas had wished for him to complete the manuscript and pursue its publication.
\medskip

\noindent\textbf{Code availability.}
The finite verification of the construction $\mathcal V_{140}$ described
in Appendix~\ref{app:v140-certificate} was performed using a short
deterministic Python program written by Bacim Alali for the revision of
this manuscript. The program is not required for the mathematical proofs:
each computation consists of a finite combinatorial check that can also
be carried out directly from the side list~\eqref{eq:sigma-sides}.
The verification code is available from Bacim Alali upon request.
\bibliographystyle{acm}
\bibliography{refs}

\end{document}